\documentclass[a4paper,UKenglish,cleveref, autoref, thm-restate]{lipics-v2019}

\usepackage{xcolor}
\usepackage{mathtools}

\nolinenumbers
\usepackage{fixme}
\fxsetup{
  status=draft,
  layout=inline,
  theme=color
}
\colorlet{fxnote}{blue}
\colorlet{fxnotebg}{yellow!75}

\newcommand{\nor}[1]{\left\|#1\right\|}
\newcommand{\norone}[1]{\left\|#1\right\|_1}

\newcommand{\ZZ}{\mathbb{Z}}
\newcommand{\ZZgeq}{\mathbb{Z}_{\geq 0}}
\newcommand{\RR}{\mathbb{R}}
\newcommand{\RRgeq}{\mathbb{R}_{\geq 0}}
\newcommand{\Ptop}{\mathcal{P}}

\DeclareMathOperator{\supp}{supp}
\DeclareMathOperator{\Conv}{Conv}
\DeclareMathOperator{\Exp}{\mathbf{E}}
\DeclareMathOperator{\hyper}{\mathsf{H}}

\newtheorem{fact}[theorem]{Fact}

\title{New Bounds for the Vertices of the Integer Hull}
\author{Sebastian Berndt}{University of Lübeck}{s.berndt@uni-luebeck.de}{https://orcid.org/0000-0003-4177-8081}{}

\author{Klaus Jansen}{Kiel University}{kj@informatik.uni-kiel.de}{}{}

\author{Kim-Manuel Klein}{Kiel University}{kmk@informatik.uni-kiel.de}{}{}
  
\authorrunning{S. Berndt and K. Jansen and K. Klein}

\Copyright{Sebastian Berndt, and Kim-Manuel Klein, and Klaus Jansen}



\begin{document}

\maketitle

 \begin{abstract}
   The vertices of the integer hull are the integral equivalent to the
   well-studied basic feasible solutions of linear programs. 
   In this paper we give new bounds on the number of non-zero
   components~---~their support~---~of these vertices matching either the best
   known bounds or improving upon them. 
   While the best known bounds make use of deep techniques, we only use basic
   results from probability theory to make use of the \emph{concentration of
     measure} effect.
   To show the versatility of our techniques, we use our results to give the
   best known bounds on the number of such vertices and an algorithm to
   enumerate them. 
   We also improve upon the known lower bounds to show that our results are
   nearly optimal. 
   One of the main ingredients of our work is a generalization of the famous
   Hoeffding bound to vector-valued random variables that might be of general
   interest. 
 \end{abstract}
\newpage
\section{Introduction}
We consider integer programs (IPs) of the form
\begin{align} \label{IP:PI}
     \min c^{\top} x; \quad 
     Ax = b; \quad 
     x \in \ZZgeq^n 
 \end{align}
for a given constraint matrix $A \in \ZZ^{m \times n}$ a right-hand side vector $b \in \ZZ^m$ and an objective vector $c \in \ZZ^n$.

Structural insights into the solution space of problems is one of the
fundamental design principle of algorithmics. Understanding that solutions of
special structure exists has led to many useful results for both theoretical and practical applications. In this work, we focus on the structure of solutions to IP (\ref{IP:PI}). Optimal solutions to the IP have a direct relation to the vertices of the integer hull of the polytope $\Ptop = \{x \in \RRgeq^n \mid Ax =b \}$. The integer hull of $\Ptop$ is defined as the convex hull of all integral points in $\Ptop$. By standard LP theory, we know that the set of optimal solution to IP (\ref{IP:PI}) always contains a vertex of the integer hull and vice versa, each vertex of the integer hull is an optimal solution for some objective $c$. Hence, vertices of the integer hull are of special interest when it comes to integer optimization. In multi-criteria optimization or in cases, when the objective $c$ is not known in advance it is useful to list all possible optimal solutions, i.e. vertices of the integer hull. Building upon a classical result by Cook et al. \cite{DBLP:journals/combinatorica/CookHKM92}, we give new bounds on the number of vertices of the integer hull as well as an algorithm to enumerate those. In many algorithmic applications, all vertices of a subpolytope need to be computed as a subroutine (see for example~\cite{goemans2014polynomiality,DBLP:conf/icalp/JansenKV16, klein2016}).

Moreover, we study the existence of optimal solutions with a possibly small number of non-zero components also called the support of a vector. In case of the relaxed linear program, it is well-known that there exist optimal solutions (i.e. the basic feasible solutions) such that their support is bounded by the number of constraints.
However, such a bound was not known for integer solutions of IP (\ref{IP:PI}) until Eisenbrand and Shmonin~\cite{eisenbrand2006caratheodory} showed a similar bound depending on the number of constraints \emph{and} the logarithm of the largest coefficient of a constraint. Their result had immediate consequences for the complexity of $NP$-hard problems (see e.\,g.~\cite{eisenbrand2006caratheodory}), for logic (see
 e.\,g.~\cite{DBLP:journals/siamcomp/KieronskiMPT14,DBLP:conf/cade/KuncakR07}), but also in the design of parameterized algorithms (see e.\,g.~\cite{DBLP:conf/esa/KnopKM17,DBLP:journals/jcss/Onn17}), approximation schemes (see
 e.\,g.~\cite{DBLP:journals/siamdm/Jansen10,DBLP:conf/icalp/JansenKV16}), and exact algorithms (see  e.\,g.~\cite{goemans2014polynomiality}). 
 More concretely, if $A$ is a $m\times n$-matrix with largest coefficient
 $\Delta=\lVert A \rVert_{\infty}$, Eisenbrand and Shmonin showed that if there
 is an integral solution $y \in \ZZgeq$ to $Ay = b$, then there also exists an
 integral solution $y'\in \ZZgeq$ with $Ay'=b$ and $|\supp(y')|\leq 2m\cdot
 \log(4m\Delta)$~\cite[Theorem 1]{eisenbrand2006caratheodory}. 
 This result was later improved by Aliev et al.~\cite{aliev2017sparse} to obtain
 a bound of $|\supp(y')|\leq
 m+\log(g^{-1}\sqrt{\det{AA^{\top}}})\leq 2m\cdot \log(2\sqrt{m}\Delta)$, where $g$ denotes the greatest common divisor of all $m\times m$ minors of $A$.
 Their result is based on a highly non-trivial variant of \emph{Siegel's lemma}
 due to Bombieri and Vaaler~\cite{bombieri1983siegel} that shows that the homogenous
 equation system $Ay=0$ has $n-m$ linearly independent integral solutions
 $y_{1},\ldots,y_{n-m}$ such that $\prod_{i=1}^{n-m}\lVert y_{i}
 \rVert_{\infty}$ is bounded by $g^{-1}\sqrt{\det{AA^{\top}}}$.
 The proof of this result makes use of different deep techniques from a variety
 of fields such as measure theory, complex analysis, and geometry of number. A
 weaker version of Siegel's lemma with a simpler proof was presented by Beck~\cite{beck2017siegel}.
 Note that both of these results only give bounds for the support of a
 \emph{feasible} integral solution of $Ay=b$.
 If one considers optimal integral solutions of the problem $\min\{c^{\top}y\mid
 Ay=b\}$, Aliev et al.~\cite{DBLP:journals/siamjo/AlievLEOW18} showed that the
 bounds of~\cite{aliev2017sparse} can be transferred to this setting:
 There is always an \emph{optimal} solution $z$  of IP (\ref{IP:PI}) with $|\supp(z)|\leq
 m+\log(g^{-1}\sqrt{\det{AA^{\top}}})\leq 2m\cdot \log(2\sqrt{m}\Delta)$. Just recently, Aliev et al.~\cite{DBLP:conf/ipco/AlievALO20} also improved upon these results for special cases. 

\paragraph*{Our Results} In this paper, we show a similar bound to Aliev et al.~\cite{DBLP:journals/siamjo/AlievLEOW18}: 
There is always an \emph{optimal} solution $z$ to IP~\eqref{IP:PI} with $|\supp(z)|\leq 2m\cdot  \log(O(\sqrt{m}\Delta))$. 
In our proof we can avoid the use of Siegel's Lemma and give a straight-forward
proof that only relies on basic results from probability theory. These tools
allow us to give a Hoeffding-type theorem to bound the deviation of a
vector-valued random variable from its expected value with regard to the $\ell_{1}$-norm. To demonstrate the usefulness of our simpler approach, we obtain several other results:
\begin{itemize}
\item In the case of $m=1$ (the \emph{knapsack polytope}), we obtain the best-known bound of $|\supp(z)|\leq \log(2\Delta \cdot \sqrt{(3/2)\cdot \log(2\Delta)})\leq \log(O(\Delta \sqrt{\log(\Delta)}))$, where the best-known bound before was $2\log(2 \Delta)\leq \log(O(\Delta^{2}))$~\cite{aliev2017sparse}.
We also give a lower bound that shows that the additive error of our result is at most $1{.}021$.
\item We show the existence of polytopes such that $|\supp(z)|\geq m\log(\Delta)+m$ for all optimal solutions $z$, where the best known lower bound before was $m\log(\Delta)^{1-\epsilon}$ for arbitrary small $\epsilon$~\cite{DBLP:journals/siamjo/AlievLEOW18}. 
 \item We show that the number of integer vertices of a polytope is bounded by
   $(n\cdot m\cdot \log(m\Delta))^{O(m\log(\sqrt{m}\Delta))}$. 
   The best known bound before was $(n\cdot m\cdot 
   \Delta)^{O(m^{2}\log(\sqrt{m}\Delta))}$~\cite{DBLP:journals/siamjo/AlievLEOW18}. 
 \item We show that the integer vertices of a polytope can be enumerated in
   running time $(n\cdot m\cdot \log(m\Delta))^{O(m\log(\sqrt{m}\Delta))}$.
   Previous results such as those of Hayes and Larman~\cite{hayes1983vertices}
   or Cook et al.~\cite{DBLP:journals/combinatorica/CookHKM92} only gave running
   times where the exponent was of order $n-1$. 
 \item We give a support-bound that not only depends on the largest entry in a constraint, but on the $\ell_{1}$-distance between the different constraints.
 We expand upon this to obtain a support-bound nearly matching the result of Aliev et al.~\cite{DBLP:journals/siamjo/AlievLEOW18} by only using Minkowski's second theorem. Due to space reasons, this can be found in the appendix in Section~\ref{sec:structure}.
 \end{itemize}

 \section{Preliminaries}
 For a matrix $A$ with columns $A_{1},\ldots,A_{n}$ and an index set $I\subseteq
\{1,\ldots,n\}$, we denote by $A[I]$ the submatrix consisting of the columns
indexed by $I$. For a matrix $A\in \mathbb{R}^{m\times n}$ and a set of vectors
$V\subseteq \mathbb{R}^{n}$, we denote by $A\cdot V=\{Av\mid v\in V\}$. 
The \emph{support} $\supp(x)$ of a vector $x\in \mathbb{R}^{n}$ is the set of
indices where $x$ is non-zero, i.\,e.~$\supp(x)=\{i\mid x_{i}\neq 0\}$.
The \emph{hypercube} $\hyper_{n}=\{0,1\}^{n}$ is simply the set of all
$\{0,1\}$-vectors of length $n$.
The $j$-th entry of a vector $x\in \mathbb{R}^{n}$ is denoted as $x[j]$.
For a set $X\subseteq \mathbb{R}^{n}$, the \emph{convex hull} $\Conv(X)$ of $X$ is
defined as the set of convex combinations of points in $x$, i.\,e.~
\begin{align*}
  \Conv(X)=\{\sum_{x\in X}\alpha_{x}x\mid \sum_{x\in X}\alpha_{x}=1, \alpha_{x} \geq 0\}. 
\end{align*}

We consider polytopes of the form $\Ptop = \{ x\mid Ax =b, x \geq 0 \}$ and say
that $x^{*}$ is a solution of $\Ptop$, if $x^{*}\in \Ptop$. 
The \emph{integer hull} $\Ptop_I$ is defined by the convex hull of all integral
points in $\Ptop$, i.\,e.~$\Ptop_I = \Conv(\Ptop \cap \ZZ^n)$. 
Obviously, an integer point $p \in \ZZ^n$ is a vertex of  $\Ptop_I$ if and only if $p$ is
a solution of the integer program $\{x\mid Ax=b, x\in \ZZgeq^{n}\}$. 
A special role is played by the \emph{vertices} of $\Ptop_{I}$. 
We say that $v\in \Ptop_{I}$ is a vertex, if it cannot be written as a convex
combination of other points in $\Ptop_{I}$, i.\,e.~$v\not\in
\Conv(\Ptop_{I}\setminus \{v\})$. 
These vertices are important as they correspond to optimal solutions of the
associated integer program~\cite{schrijver1998theory}.
\begin{fact}
  \label{fact:optimal}
  For any objective function $c\in \ZZ^{n}$ there is a vertex $v\in \Ptop_{I}$ that is an
  optimal solution to the integer program $\min\{c^{\top}x\mid Ax=b, x\in
  \ZZgeq^{n}\}$, if the integer program has a finite optimum. 
\end{fact}

It is easy to see that a point $p\in \Ptop_{I}$ is not a vertex of the integer
hull if there exists a non-zero solution of a certain integer program. 

\begin{lemma}
  \label{lem:eisenbrand:shmonin}
  If $v$ is a vertex of the integer hull $\Ptop_{I}$, the following system does not have a
  non-zero solution:
\begin{align} \label{IP:kernel}
    A x = 0; \quad   
    -1 \leq\  &x_i \leq 1\ \ \forall i \in \supp(v); \quad 
    x_{i} = 0\ \ \forall i\not\in \supp(v); \quad 
    x \in \ZZ^{n}
  \end{align}
  \end{lemma}
  \begin{proof}
    Suppose there is a non-zero solution $x$ of program~\eqref{IP:kernel}, then $v-x$ and
    $v+x$ are solutions of $\{x\mid Ax=b, x\in \ZZgeq^{n}\}$  and therefore $v-x$ and $v+x$ are
    feasible points of $\Ptop_I$. 
    As $v$ can be written as the convex combination $v = \frac{v-x}{2} +
    \frac{v+x}{2}$, the point $v$ is not a vertex of the integer hull $\Ptop_I$.
  \end{proof}

\section{The Knapsack Polytope}
\label{sec:knapsack}
In this section, we consider the special case that the matrix $A$ consists of a single
row, i.e. $A = a^{\top}$ for some non-zero vector $a \in \ZZ^n$. 
The right-hand side $b\in \ZZ$ is thus a single integer. 
In this case, $\Ptop_I =
\Conv( \{ x \mid a^{\top} x = b, x\geq 0\} \cap \ZZ^n)$ is called the knapsack polytope
and the vector $a$ corresponds to the sizes of the respective items (note that
we also allow negative entries). 
Using Lemma~\ref{lem:eisenbrand:shmonin} in combination with a pigeonhole
argument, one can easily show that the support of vertices in $\Ptop_I$ is
bounded (see \cite{eisenbrand2006caratheodory}). To see this, suppose that a point $p \in \Ptop_I$ has a
support with cardinality $s:= |\supp(p)|$ with $s > \log{\norone{a}} +1$ and
hence $|\hyper_{s}| = 2^s > \norone{a}+1$. However, looking at all $x\in \hyper_{s}$, the sum
$\sum_{i=1}^n a_i x_i$ can take at most $\norone{a}+1$ different values. By the pigeonhole principle, there exist two different points $x', x'' \in \hyper_s$ with $a^{\top} x' = a^{\top} x''$. However, this means that $x''-x'$ is a non-zero solution of IP \eqref{IP:kernel}. By Lemma~\ref{lem:eisenbrand:shmonin}, the point $p$ can hence not be a vertex.

\subsection{Characterizing the Vertices}
In order to improve upon the above technique, we elaborate on the observation
that the value of $a^{\top} x$ for most $x \in \hyper_{n}$ is around $\sum_{i=1}^n a_i /2$.
In other words, if we choose a point $x \in \hyper_{n}$ randomly, then the value of
$a^{\top} x$ is likely to be close to the expected value, which is $\sum_{i=1}^n a_i
/2$. The improved bound for the support can then be derived by using the above
pigeonhole argument on a smaller area around the expected value. In other words
we choose a subset $\hyper'_{n}$ of the hypercube $\hyper_{n}$ that maps onto the smaller
area $a^{\top}\cdot \hyper'_{n}$ around the expected value. To choose this subset, we make use
of the well-known \emph{Hoeffding bound}.
\begin{theorem}[Theorem 4.12 in \cite{mitzenmacher2017probability}]
  \label{thm:hoeffding}
Let $X_{1},\ldots,X_{s}$ be independent random variables such that $\ell_{i} \leq
X_{i}\leq u_{i}$ for all $i$. Let $X=\sum_{i=1}^{s}X_{i}$ be its sum. 
 For all
$\delta > 0$, we have
\begin{align*}
  \Pr[|X - \mu| \geq \delta] \leq 2\exp(-(2\delta^{2})/\sum_{i=1}^{s}(u_{i}-\ell_{i})^{2}), 
\end{align*}
where $\mu$ is the expected value of $X$.
\end{theorem}

Using those ideas, we obtain the following support bound for vertices of the
integer hull of a knapsack polytope.
\begin{theorem}
  \label{thm:kp:supp_bound}
  For each vertex $v$ of the integer hull of the knapsack polytope $\Ptop_{I}$ we have that
  $|\supp(v)| \leq \log(3{.}51\cdot \nor{a[\supp(v)]}_2)$. 

\end{theorem}
\begin{proof}
  We show that for each $x \in \Ptop_{I}$, the inequality $\supp(x) >
  \log(3{.}51\cdot \nor{a[\supp(x)]}_2)$ implies that x is not a vertex.  Consider any integral solution $x^{*}$ of $\{x\mid a^{\top}x=b, x\in \ZZgeq^{n}\}$ with support
$S=\supp(x^{*})$ of size $s=|S|$. 
Let $a[S]$ be the subvector of $a$ projected to the elements indexed by $S$.
  
\textbf{The random process:} Note that choosing a random point in $\hyper_s$ is
equivalent to the following: For each $i \in \supp(x^*)$, consider the random
variable $X_{i}$ that is equal to $0$ with probability $1/2$ and equal to $a[i]$
with probability $1/2$. Let $X_{S}=\sum_{i\in S}X_{i}$ be the sum of these
variables. The fact that $X_{S}$ does not deviate much from its expected value is
directly implied by Theorem~\ref{thm:hoeffding}.


\textbf{Construct $\hyper'_{s}$:} 
To use the concentration of measure approach, we choose $\hyper'_{s}=\{x\in
\hyper_{s}: 
| a[S]\cdot x - \sum_{i\in S}a[i] /2| \leq \sqrt{\lVert a[S]
  \rVert_{2}}\}$. 


\textbf{Analyze $|\hyper'_{s}|$: }As all $X_{i}$ are independent and the expected value of $X_{S}$ is $
\mu_{S} := \sum_{i\in S}a[i] / 2$,  Theorem~\ref{thm:hoeffding} implies that 
\begin{align*}
  \Pr[|X-\mu_{S}| \geq \delta]\leq
  2\exp(-(2\delta^{2})/\lVert a[S] \rVert^{2}_{2}). 
\end{align*}
Hence, for $\delta=\alpha\cdot \lVert a[S] \rVert_{2}$ for some $\alpha\in
\mathbb{R}_{> 0}$, we have
$ \Pr[|X_{S}-\mu_{S}| \geq \alpha\cdot \lVert a[S] \rVert_{2}]\leq
2\exp(-2\alpha^{2})$. 
This directly implies that $|\hyper'_{s}|\geq (1-2\exp(-2\alpha^{2}))\cdot 2^{s}$.

\textbf{Analyze $|a^{\top}\cdot \hyper'_{s}|$:} On the other hand, $a[S]\cdot \hyper'_{s}\subseteq \left[\mu_{S} -
\alpha \lVert a[S] \rVert_{2}, \mu_{S}+\alpha\lVert a[S]
  \rVert_{2}\right]$ and thus $|a^{\top}\cdot \hyper'_{s}|\leq 2\alpha\lVert a[S]
\rVert_{2}$.

\textbf{Compare $|\hyper'_{s}|$ and $|a^{\top}\hyper'_{s}|$:}
Hence, if $(1-2\exp(-2\alpha^{2}))\cdot 2^{s} > 2\alpha\lVert a[S]
\rVert_{2}$, we have $|\hyper'_{s}|> |a^{\top}\cdot \hyper'_{s}|$ and can thus use
Lemma~\ref{lem:eisenbrand:shmonin}. Solving for $s$ gives
\begin{align*}
  s > \log\left[ \frac{2\alpha}{(1-2\exp(2-2\alpha^{2}))} \cdot \lVert a[S] \rVert_{2} \right]. 
\end{align*}
Minimizing the function $\alpha\mapsto \frac{2\alpha}{(1-2\exp(2-2\alpha^{2}))}$ on
the positive reals shows that it has a minimum smaller than $3{.}51$ at
$x=\frac{1}{2}\sqrt{-2W_{-1}\left( -\frac{1}{4\exp(5/2)} \right)-1}\approx 1{.}5986$, where
$W_{-1}$ is the negative real branch of the Lambert $W$ function defined by
$W_{-1}(x\exp(x))=x$ for $x\leq -1$. 
Choosing $\alpha\approx 1{.}5986$ accordingly gives the simpler inequality $s > \log(3{.}51\cdot
\lVert a[S] \rVert_{2})$. 

Hence, for each $x\in \Ptop_{I}$,  the inequality $|\supp(x)| > \log(3{.}51\cdot
\lVert a[\supp(x)] \rVert_{2})$ implies that $x$ is not a
vertex. 
\end{proof}

To compare this bound with the known bounds in literature, let $\Delta=\lVert a
\rVert_{\infty}$. As $\lVert a[S] \rVert_{2}\leq \sqrt{|S|}\cdot \Delta$, we
know that  $s > \log(3{.}51\cdot \sqrt{s}\cdot \Delta)$ implies $|\hyper'_{s}|> |A\cdot
\hyper'_{s}|$. We thus have the following lemma:
\begin{lemma}
  \label{lem:kp:supp_bound}
  For each vertex $v$ of the integer hull of the knapsack polytope $\Ptop_{I}$ with  $|\supp(v)|=s$, we have $s \leq
  \log(3{.}51\cdot \sqrt{s}\cdot \Delta)$. 
\end{lemma}

As $\Delta$ denotes the largest absolute value of a number in $a$, and a vertex
cannot use variable corresponding to a number $k$ and a variable corresponding
to a number $-k$ at the same time,  we clearly have $s\leq \Delta$ and
can thus obtain the following simple corollary:
\begin{corollary}
  \label{cor:kp:supp_bound}
  For each vertex $v$ of the integer hull of the knapsack polytope $\Ptop_{I}$ with
  $|\supp(v)|=s$, we have $s \leq (3/2)\log(2{.}4\cdot \Delta)$. 
\end{corollary}
Note that the best known bound before was due to Aliev et
al.~\cite{DBLP:journals/siamjo/AlievLEOW18} and was $2\log(2\Delta)$.
In Section~\ref{sec:inequality}, we provide an analysis to obtain tighter
bounds from inequalities such as those of Theorem~\ref{thm:kp:supp_bound} depending
on the relation between $\Delta$ and $m$. For example, the results of
Theorem~\ref{thm:ineq:main} imply that a bound of of $s\leq \log(2\Delta\cdot
\sqrt{(3/2)\cdot \log(2\Delta)})$ can be constructed also.

\paragraph*{A Lower Bound on the Support}

Here, we present a simple lower bound for the support of a vertex of $\Ptop_I$.

Consider the polytope $\Ptop = \{x \mid a^{\top} x = b, x \geq 0 \}$ with
$a^{\top} = (2^0, 2^1, \ldots , 2^{d-1})$ and $b = 2^d-1$. Clearly, $v^{\top} =
(1, \ldots , 1)$ is a vertex of $\Ptop_I$ as $v$ is the unique optimal integral
solution that maximizes the objective function $c^{\top} = (2^0, 2^1, \ldots ,
2^{d-1})$. For the support of $v$, we obtain $|supp(v)| = n = d=\log(\Delta)+1$. 
Furthermore, we have $\lVert a
\rVert_{2}=\sqrt{\sum_{i=0}^{d-1}2^{2i}}=\sqrt{\sum_{i=0}^{d-1}4^{i}}=\sqrt{(4^{d}-1)/3}$ and hence
$  \log(\lVert a \rVert_{2}) = (1/2)\log([4^{d}-1]/3)
  \leq (1/2)\log(4^{d}/(3-\epsilon)) = d-(1/2)\log(3-\epsilon)$ for a
  sufficiently small $\epsilon \leq 3/(4^{d}-1)$. 
  We thus obtain the following lemma:
  \begin{lemma}
    \label{lem:kp:lower_bound}
    There is an integer hull of a knapsack polytope $\Ptop_{I}$ with a vertex $v$ with 
      $|\supp(v)| \geq \log(\lVert a \rVert_{2})+(1/2)\log(3-\epsilon)$
      for $\epsilon \leq 3/(4^{d}-1)$. 
  \end{lemma}
  Theorem~\ref{thm:kp:supp_bound} thus gives us an upper bound of
  $\log(3{.}51\cdot \lVert a \rVert_{2})\leq \log(\lVert a \rVert_{2})+1{.}811$,
  while Lemma~\ref{lem:kp:lower_bound} gives an lower bound converging to
  $\log(\lVert a \rVert_{2})+(1/2)\log(3) \geq \log(\lVert a
  \rVert_{2})+0{.}79$. The additive error of Theorem~\ref{thm:kp:supp_bound} is
  thus at most $1{.}021$.

\subsection{Bounding the Number of Integer Vertices}
\label{sec:kp:number}

The result of Corollary~\ref{cor:kp:supp_bound} directly shows that the number of
vertices of the integer hull of a knapsack polytope $\Ptop_{I}=\Conv(\{x\mid a^{\top}x = b,
x\geq 0\}\cap \ZZ^{n})$ is at most $\binom{n}{\ell}\cdot (b+1)^{\ell}\leq (n\cdot (b+1))^{\ell}$, where
$\ell=(3/2)\log(2{.}4\cdot \Delta)$, 
but we can significantly improve on this.
To do so, we show that Corollary~\ref{cor:kp:supp_bound} can be easily combined
with the techniques from Hayes and Larman~\cite{hayes1983vertices}.

First, we get rid of the dependency on $b$ completely by using well-known
results on the proximity of integral solutions and fractional solutions, as
shown in~\cite{DBLP:journals/siamjo/AlievLEOW18}.
Fix any vertex $v$ of $\Ptop_{I}$. 
Now, as $v$ is a vertex, there is $c\in \ZZ^{n}$ such that $v$ is the unique
optimum solution of $\min \{c^{\top}x\mid ax=b, x\in \ZZ^{n}_{\geq 0}\}$. 
Now consider an optimal basic feasible solution $y$ of the relaxed program $\min
\{c^{\top}x\mid ax=b, x\in \RR^{n}_{\geq 0}\}$. The proximity result
of~\cite{DBLP:journals/talg/EisenbrandW20} then shows that one can bound 
$  \lVert v-y \rVert_{1}\leq 2\Delta+1$. 
Hence, every vertex solution of $\Ptop_{I}$ is near a basic feasible solution of the
relaxed program.
Fix some basic feasible solution $y$ of the relaxed program.
Let $C_{y}=\{z\in \Ptop_{I}\mid \lVert y-z \rVert_{1}\leq 2\Delta+1\}$ be the
surrounding box around $y$. 
We partition $C_{y}$ into smaller boxes such that each box contains at most one
vertex. 
Let $L\in \ZZ^{n}$ be the lower corner of $C_{y}$, i.\,e.~$L_{i}=\max\{\lceil
y_{i}-(2\Delta+1) \rceil,0\}$. 
Now, let $d = \lceil \log(8(2\Delta+1))\rceil$. 
We now construct exponentially growing intervals $I_{j}$ with $I_{0}=[0,1)$ and
$I_{j}=[2^{j-1},2^{j})$ for $j\geq 1$.
For a vector $\vec{k}=(k[1],\ldots,k[n])\in \ZZ^{n}$ with $k[i]\leq d$,
we define the \emph{Box} $B_{\vec{k}}\subseteq \ZZ^{n}$ with
$B_{\vec{k}}=L+\bigtimes_{i=1}^{n}I_{k[i]}$.
Let $\mathcal{B}'=\{B_{\vec{k}}\mid \vec{k}\in \ZZ^{n}\mid k[i]\leq d\}$
the set of all of these boxes and $\mathcal{B}\subseteq \mathcal{B}'$ be the set
of all boxes that intersect $C_{y}$. 
Clearly, all of these boxes are disjoint by construction. Furthermore, $C_{y}\subseteq
\bigcup_{B\in \mathcal{B}}B$:
Consider any $z\in C_y$. 
By definition, we have $\lVert y-z \rVert_{1}\leq 2\Delta+1$. 
Hence, there is a vector $v\in \ZZgeq^n$ such that $L+v=z$, where $\lVert v\rVert_{\infty}\leq 2(2\Delta+1)$.
Hence, for each $i=1,\ldots,n$, there is an integer $k[i]$ such that $v_i\in I_{k[i]}$. 
As $v_i\leq 2(2\Delta+1)$, we have $k[i]\leq \log(4(2\Delta+1)) < \lceil \log(8(2\Delta+1)) \rceil  = d$. 
Hence, $z\in B_{\vec{k}}$ and thus $C_y\subseteq
\bigcup_{B\in \mathcal{B}}B$. 
We will now show that a box containing a vertex
cannot contain any more integral points. 

\begin{lemma}
  \label{lem:kp:boxes}
  No box $B\in \mathcal{B}$ containing a vertex $v$ can contain another
  integral point of $\Ptop_{I}$. 
\end{lemma}
\begin{proof}
  Suppose that some box $B=B_{\vec{k}}$ contains a vertex $v$ and another
  integral point $p$.
  We will argue that $q=2v-p$ is also an integral point of $\Ptop_{I}$ which
  contradicts the fact that $v$ is a vertex, as $v=(p+q)/2$. 
  It is easy to see that $q$ is integral and that
  $a^{\top}q=2a^{\top}v-a^{\top}p=2b-b=b$ holds. 
  We only need to show that $q\geq 0$.
  For indices $i$ with $v[i]\geq p[i]$, this clearly holds.
  Note that $v[i]=p[i]$ holds for all $i$ with $k[i]=0$, as $I_{0} = [0,1)$ and thus contains only a single integer.  
  The only remaining case to consider is an index $i$ with $v[i] < p[i]$ and
  $k[i] > 0$.
  As $v,p\in B_{\vec{k}}$, we know that $v[i],p[i]\in I_{k[i]}$ and thus $v[i]
  > p[i]/2$. We can thus conclude that $q[i] = 2 v[i] - p[i]  \geq 0$.
\end{proof}

Lemma~\ref{lem:kp:boxes} directly shows that the number of integer vertices in
box 
$C_{y}$ is at most $d^{n}$, i.\,e.~at most
$\lceil \log(8(2\Delta+1)\rceil ^{n}$, as there are at most this many boxes in
$\mathcal{B}'$. 
But Corollary~\ref{cor:kp:supp_bound} shows that all boxes $B_{\vec{k}}$ with
$|\supp(\vec{k})| > \ell$ can also not contain any
integer vertices.
The number of vectors $\vec{k}\in \ZZ^{n}$ with (i) $k[i] \leq d$ and (ii)
$|\supp(\vec{k})| \leq  \ell$ is bounded by
\begin{align*}
  \sum_{j=1}^{\ell}\binom{n}{j}\cdot d^{j}  < 
  \ell \cdot 
  n^{\ell}\cdot d^{\ell}\leq 
 \ell\cdot (n\cdot \lceil \log(8(2\Delta+1)) \rceil )^{\ell}.
\end{align*}
Hence, there are at most $\ell\cdot (n\cdot
\lceil \log(8(2\Delta+1)) \rceil )^{\ell}$ vertices in $C_{y}$. 
As there are at most $n$ basic feasible solutions of the relaxed program, the total
number of vertices of the integer hull of a knapsack polytope is at most
$\ell\cdot  (n\cdot \lceil \log(8(2\Delta+1)) \rceil)^{\ell+1}$. 

\begin{theorem}
  \label{thm:kp:vertices}
  The number of vertices of the integer hull of a  knapsack polytope $\Ptop_{I}$ is
  at most $(n\cdot \log(\Delta))^{O(\log(\Delta))}$. 
\end{theorem}

\subsection{Enumerating the Vertices}
\label{sec:kp:enumerating}
While Theorem~\ref{thm:kp:vertices} gives us an upper bound on the number of
vertices of $\Ptop_{I}$, it does not directly lead to an algorithm enumerating
all of them. 
As above, we split the polytope into boxes $C_{y}$ and then into smaller boxes
$\mathcal{B}$ and define $\ell=(3/2)\log(2{.}4\cdot \Delta)$. 
Lemma~\ref{lem:kp:boxes} shows that if a box $B\in \mathcal{B}$ contains two
integral points, it does not contain a vertex. 
We will now use this lemma to algorithmically enumerate the vertices.

\begin{theorem}
\label{thm:knapsack:number}
  All integer vertices of the integer hull of the  knapsack polytope $\Ptop_{I}$ can be enumerated in time
    $(n\cdot \log(\Delta))^{O(\log(\Delta))}$. 

\end{theorem}
\begin{proof}

First, note that Corollary~\ref{cor:kp:supp_bound} shows that no box $B_{\vec{k}}$
with $|\supp(\vec{k})| > \ell$ contains any vertex. 
In the following, we thus iterate through all boxes $B_{\vec{k}}$ with
$|\supp(\vec{k})|\leq \ell$ and filter out all boxes containing
either at least two integral points or none. 
For all the remaining boxes that contain exactly one integral point, we check
whether this single point is a vertex.

Fix some box $B_{\vec{k}}$ with $|\supp(\vec{k})|\leq \ell$. 
We first use the classical algorithm of Lenstra~\cite{lenstra1983integer} and
Kannan~\cite{kannan1987minkowski} to check whether $B_{\vec{k}}$ contains an 
integral point. 
If no such integral point exists, we discard the box. 
If an integral solution $x^{*}$ exists, we can search for another integral
solution $x$ where we additionally force $x_{i} \leq x^{*}_{i} -1$
resp.~$x_{i}\geq x^{*}_{i}+1$ for each $i$ individually. 
If any other integral solution exists, we also discard the box. 
In total, we make at most $2|\supp(\vec{k})|+1$ calls of the algorithm of
Lenstra and Kannan. 
This can be done in time $|\supp(\vec{k})|^{O(|\supp(\vec{k})|)}\cdot
\log(\Delta)^{O(1)}\leq \log(\Delta)^{O(\log \Delta)}$. 
If $B$ only contains a single solution, we still need to check whether this is a
vertex. 
This can be done as in Hayes and Larman~\cite{hayes1983vertices}.
Let $\mathcal{B}'\subseteq \mathcal{B}$ be the set of boxes containing exactly
one integral solution and let $W$ be these solutions.
The proof of Theorem~\ref{thm:knapsack:number} shows that $W$ contains all vertices. 
Now, $w\in W$ is a vertex iff $w\not\in \Conv(W\setminus \{w\})$. Hence, to
determine whether $w$ is a vertex, we just need to check whether the following
linear program (with variables $\lambda_{w}$) has a solution:
\begin{align*}
  w = \sum_{w'\in W\setminus \{w\}}\lambda_{w'}w'; \quad 
  \sum_{w'\in W\setminus \{w\}}\lambda_{w'} = 1; \quad 
  \lambda_{w'}\geq 0 \  \forall w'\in W\setminus \{w\}
\end{align*}
This linear program can be solved in time $|W|^{O(1)}\cdot
\log(\Delta)^{O(1)}\leq (n\log(\Delta))^{O(\log \Delta)}$ via the ellipsoid
algorithm~\cite[Theorem (6.6.5)]{grotschel2012geometric}, as Theorem~\ref{thm:knapsack:number} shows that $|W|\leq (n\log(\Delta))^{O(\log \Delta)}$.

In total, the complete running time to enumerate the vertices is
   $(n\cdot \log(\Delta))^{O(\log(\Delta))}$.
\end{proof}

\section{Handling General Polytopes}
\label{sec:general}
In this section, we generalize the results of Section~\ref{sec:knapsack}. We consider the integer hull $\Ptop_I = \operatorname{Conv}(\Ptop \cap \ZZ^n)$ for polytopes $\Ptop = \{Ax=b, x\in \RRgeq^{n}\}$ where $A$ is an integral
$m\times n$-matrix consisting of 
rows $a^{(1)},\ldots,a^{(m)}$. 

In this case, the pigeonhole argument can be applied in the same way: Assume
that a point $x \in \Ptop_I$ has support with cardinality $s := |\operatorname{supp}(x)|$ with
$s > m \log \Delta + 1$ and hence $|\hyper_s| = 2^s > 2\Delta ^{m}$. Note that the 
sum of column vectors $\sum_{i\in S} A_i x_i$ can have at most
$\prod_{i=1}^{m}(\lVert a^{(i)} \rVert_{1}+1)\leq (s\cdot \Delta+1)^{m}$
different values for $x\in \hyper_{s}$.
Hence, if $2\Delta^m > (s\cdot \Delta+1)^m$, we have $|\hyper_s| > |\hyper_s\cdot A_S|$ and thus two different points $x,x'\in \hyper_s$ with $A_{S} x = A_{S}x'$.
Lemma~\ref{lem:eisenbrand:shmonin} thus implies a bound of $|\supp(v)|\leq m\cdot
\log(|\supp(v)\cdot \Delta+1)$ for each vertex $v$. 

A simple approach to generalize Theorem~\ref{thm:kp:supp_bound} is the use of
the \emph{union bound} to handle all $m$ constraints simultaneously.
This would introduce an additional term of $\log(m)$ and the bound on the
support would be roughly $|\supp(v)| \leq O(\log(\prod_{j=1}^{m}[\sqrt{\lVert
  a^{(j)} \rVert_{2}} \cdot \log(m)]))$.
In the following, we will show that we can actually get rid of this $\log(m)$
term. 

\subsection{Characterizing the Vertices}
The general strategy for the proof of the bound on the support of general
polytopes is similar to the proof of Theorem~\ref{thm:kp:supp_bound}: 
We again observe that the value of $A\cdot x$ for most $x\in \hyper_{n}$ is
centered around $\sum_{i=1}^{n}A_{i}/2$.
In other words, if we choose a point $x \in \hyper_{n}$ randomly, then the value
of $A\cdot x$ is likely to be close to the expected value, which is
$\sum_{i=1}^{n}A_{i}/2$. 
The improved bound for the support can then be derived by using the above
pigeonhole argument on a smaller area around the expected value. 
In other words we choose a subset $\hyper'_{n}$ of the hypercube $\hyper_{n}$
that maps onto the smaller area $A\cdot \hyper'_{n}$ around the expected value. 
For the knapsack polytope, we used Theorem~\ref{thm:hoeffding}~---~the
Hoeffding bound~---~to construct this subset. 
Unfortunately, this bound is not guaranteed to hold for vector-valued random
variables. 
We thus prove the following Hoeffding-type theorem for vector-valued random
variables and postpone its proof to Section~\ref{sec:measure}. 
\begin{restatable}{theorem}{hoeffdingvector}
  \label{thm:hoeffding:vector}
  Let $Y_{1},\ldots,Y_{n}$ be independent random variables that
each output a $n$-dimensional vector with $\Pr[Y_{i}[j]\in
[\,\ell_{i}[j],u_{i}[j]\,]]=1$. Furthermore, let $Y=\sum_{i}Y_{i}$ and $\mu :=
\Exp[Y]$. 
  Then, for every $\delta > 0$, we have 
  \begin{align*}
&\Pr\left[\lVert Y-\mu \rVert_{1} \geq   1{.}12\cdot  \sum_{j=1}^{m}\sqrt{\sum_{i=1}^{n}(\ell_{i}[j]-u_{i}[j])^{2}} + \delta\right]\leq\\ &\quad \quad 2\exp(-(2\delta^2)/\sum_{i=1}^{n}[\sum_{j=1}^{m} (u_{i}[j]-\ell_{i}[j])]^{2}).
\end{align*}
\end{restatable}

Using those ideas, we obtain the following support bound for vertices of the
general integer hull.
\begin{theorem}
  \label{thm:general:supp_bound}
  For each vertex $v$ of the integer  hull $\Ptop_{I}$ with $s=|\supp(v)|$, we have that 
      $|\supp(v)| \leq m\cdot \log(2e\Gamma/m+2e)$,
    where $\Gamma=1{.}12\sum_{j=1}^{m}\lVert
a^{(j)} \rVert_{2}+ \sqrt{\sum_{i=1}^{s}\lVert A_{i}
  \rVert^{2}_{1}}$, and $A_{1},\ldots,A_{s}$ are the columns of $A[\supp(v)]$ and
$a^{(1)},\ldots,a^{(m)}$ are the rows of $A[\supp(v)]$ and $e=\exp(1)$ is Euler's
number. 
\end{theorem}
 
 \begin{proof}
  We show that for each $x\in \Ptop_{I}$, the inequality $\supp(x) >
  m\cdot \log(2e\Gamma/m+2e)$ implies that x is not a vertex. 
Consider any integral solution $x^{*}$ of $\{x\in \ZZgeq^{n}\mid Ax=b\}$ with support
$S=\supp(x^{*})$. 
Let $A_{1},\ldots,A_{s}$ be the columns of matrix $A[S]$, where $s=|S|$, and
$a^{(1)},\ldots,a^{(m)}$ be its rows.

\textbf{The random process:}
We will show that the output of the random process of choosing each
column with probability $1/2$ independently will not likely deviate from $L := (1/2)\sum_{i=1}^{s}A_i$.
For $i=1,\ldots,s$, we consider the random variable $Y_{i}$ which is equal to
$A_{i}/2$ with probability $1/2$ and equal to $-A_{i}/2$ with probability $1/2$.
Let $Y=\sum_{i} Y_{i}$ be the sum of these random variables. 
Note that $Y+L$ is exactly the random process where each column is chosen with
probability $1/2$. 

\textbf{Construct $\hyper'_{s}$:}
To use the  concentration of measure approach, we choose $\hyper'_{s}=\{x\in
\hyper_{s}: \lVert L-A[S]x \rVert_{1}\leq \Gamma\}$. 

\textbf{Analyze $|\hyper'_{s}|$: }
Now, using Theorem~\ref{thm:hoeffding:vector} and choosing
$\delta=\sqrt{\sum_{i=1}^{n}[\sum_{j=1}^{m} (u_{i}[j]-\ell_{i}[j])]^{2}}$, 
where $u_{i}=A_{i}$ and $\ell_{i}=-A_{i}$, 
we obtain 
\begin{align*}
&\Pr\left[\lVert Y \rVert_{1}\geq  1{.}12\cdot  \sum_{j=1}^{m}(\sum_{i=1}^{n}(\ell_{i}[j]-u_{i}[j])^{2})+ \sum_{i=1}^{n}[\sum_{j=1}^{m} (u_{i}[j]-\ell_{i}[j])]^{2}\right] \leq\\
& 2\exp(-2)\leq 0{.}28.
\end{align*}
Note that the expected value $\mu$ of $Y$ is the all-zero vector. 
As indicated above, this shows that $|\hyper'_{s}|\geq 0{.}72\cdot 2^{s}$.

\textbf{Analyze $|A\cdot \hyper'_{s}|$:}
Note that to upper bound the number of vectors $x\in \hyper_{s}$ with $\lVert L-A[S]x \rVert_{1}\leq
\Gamma$, we can count the integral vectors $y\in \mathbb{Z}^{m}$
 with
$\lVert y \rVert_{1}\leq \Gamma$. 
There are at most $\binom{\Gamma+m}{m}$ non-negative solutions
$x=(x_{1},\ldots,x_{m})$ to the inequality $\sum_{i=1}^{m}x_{i}\leq \Gamma$: Write
$\Gamma$ in unary and choose $m$ separating symbols. Hence the number of integral vectors $y$ with $\lVert y
\rVert_{1}\leq \Gamma$ is strictly less than 
$2^{m}\cdot \binom{\Gamma+m}{m}$, as vectors containing a zero are counted
multiple times. 
 As $\binom{m}{k}\leq \left( \frac{m\cdot e}{k}
\right)^{k}$, where $e$ is Euler's number, the
number of such vectors is strictly less than
\begin{align*}
  2^{m}\left( \frac{e (\Gamma+m)}{m} \right)^{m} =(2e\Gamma/m + 2e)^{m}.
\end{align*}

\textbf{Compare $|\hyper'_{s}|$ and $|a^{\top}\hyper'_{s}|$:}
If $|\hyper'_{s}| > |A[S]\cdot \hyper'_{s}|$,
we can again use Lemma~\ref{lem:eisenbrand:shmonin}. 

Hence, as $|A[S]\cdot \hyper'_{s}|\leq (2e \Gamma/m+2e)^{m}$ and
$|\hyper'_{s}|\geq 0{.}72\cdot 2^{s}$, for each $x\in \Ptop_{I}$, the inequality
$|\supp(x)| > m\cdot \log(2e\Gamma/m+2e)$ implies that $x$ is not a vertex. 
\end{proof}

To compare this bound with the known bounds in literature, let $\Delta=\lvert A
\rvert_{\infty}$. 
Now, for each row $a^{(j)}$ of $A[S]$, we have $\lVert a^{(j)} \rVert_{2}\leq
\sqrt{s}\Delta$. 
Hence, the sum of these norms is at most $m\sqrt{s}\Delta$. 
Furthermore, for each column $A_{i}$ of $A[S]$, we have $\lVert A_{i}
\rVert^{2}_{1}\leq m^{2}\Delta^{2}$. 
Summing up these norms and taking the square root thus gives
$\sqrt{\sum_{i=1}^{n}\lVert A_{i} \rVert^{2}_{1}}\leq m\sqrt{s}\Delta$.
Hence, $\Gamma\leq 2{.}12\cdot m\cdot \sqrt{s}\cdot \Delta$. 
We thus know that  $s > m\log(4{.}24\cdot e\sqrt{s}\Delta+2e)$ implies $|\hyper'_{s}|> |A\cdot
\hyper'_{s}|$. We thus have the following lemma:
\begin{lemma}
  \label{lem:general:supp_bound}
  For each vertex $v$ of the integer  hull $\Ptop_{I}$ with  $|\supp(v)|=s$,
  we have $s \leq m\log(4{.}24\cdot e\sqrt{s}\Delta+2e)$. 
\end{lemma}

Using a similar approach to Eisenbrand and
Shmonin~\cite{eisenbrand2006caratheodory}, we can obtain a more useful bound (see Section~\ref{sec:inequality:warmup}).  
\begin{corollary}
  \label{cor:general:supp_bound}
  For each vertex $v$ of the integer  hull $\Ptop_{I}$ with  $|\supp(v)|=s$,
  we have $s \leq 2m\log(24 \sqrt{m} \Delta)$. 
\end{corollary}
A more refined version of this bound and a corresponding proof for it is given
in Lemma~\ref{lem:ineq:supp_bound}.
Again, in Section~\ref{sec:inequality}, we provide an analysis to obtain tighter
bounds from inequalities such as those of Theorem~\ref{thm:general:supp_bound} depending
on the relation between $\Delta$ and $m$.

\paragraph*{A Lower Bound for General Polytopes}
The lower bound from the previous section for $m=1$ can easily be generalized to
arbitrary $m$.
By this, we improve upon the lower bound given by Aliev et
al.~\cite{DBLP:journals/siamjo/AlievLEOW18} which showed that a support of
at least $m\log(\Delta)^{1/(1+\epsilon)}$ is always possible for all $\epsilon > 0$.

Consider the polytope $\Ptop = \{x \mid Ax = b, x \geq 0 \}$ with right hand side vector $b^T = (2^d -1, \ldots , 2^{d}-1)$ and constraint matrix 
\begin{align*}
    A = \begin{pmatrix} 2^0 & \cdots & 2^{d-1} & 0 & & \cdots&  & & 0 \\
    0 & \cdots & 0 & 2^0 & \cdots & 2^{d-1} & 0 & \cdots &0 \\
    & \vdots & & & &  & \ddots \end{pmatrix},
\end{align*}
which contains $m$ copies of the vector $(2^0, \ldots 2^{d-1})$ on the diagonal
line. Let $a^{(1)},\ldots,a^{(m)}$ be the rows of this matrix and
$A_{1},\ldots,A_{n}$ be its columns. 
Clearly, $v^T = (1, \ldots , 1)$ is a vertex of $\Ptop_I$ as it is the unique
optimal integral solution that maximizes the objective vector $c^T = (2^0,
\ldots , 2^{d-1}, 2^0 , \ldots , 2^{d-1} , \ldots )$. 
The support of $v$ is obviously exactly $n=m\cdot d = m \log(\Delta) + m$.

  \begin{lemma}
    \label{lem:general:lower_bound}
    There is an integer hull $\Ptop_{I}$ with a vertex $v$ with 
    $|\supp(v)| \geq m\log(\Delta)+m$. 
    \end{lemma}

Now, consider the support bound of Theorem~\ref{thm:general:supp_bound} that
guarantees a solution with support at most
\begin{align*}
  m\cdot \log\left( 2e\left[ 1{.}12\sum_{j=1}^{m}\lVert
a^{(j)} \rVert_{2}+ \sqrt{\sum_{i=1}^{s}\lVert A_{i}
  \rVert^{2}_{1}}) \right]/m +2e \right). 
\end{align*}
As shown in the lower bound described in 
Section~\ref{sec:knapsack}, the $\ell_{2}$-norm of each row is at most
$2^{d}$
and thus $\sum_{j=1}^{m}\lVert a^{(j)} \rVert_{2}\leq m\cdot 2^{d}$. Furthermore, we
have $\sqrt{\sum_{i=1}^{n}\lVert A_{i} \rVert_{1}^{2}}=\sqrt{m\cdot
  (4^{d}-1)/3}\leq \sqrt{m}\cdot 2^{d}$. Hence,
Theorem~\ref{thm:general:supp_bound} guarantees a solution with support at most
\begin{align*}
  &m\cdot \log(2\cdot e\cdot (1{.}12\cdot m\cdot 2^{d}+\sqrt{m}2^{d})/m+2e)=\\
  &m\cdot \log(2\cdot e\cdot (1{.}12\cdot 2^{d}+2^{d}/\sqrt{m})+2e)\leq\\
  &m\cdot \log(6e2^{d})=m\cdot d\cdot \log(2\sqrt[d]{6e}).
\end{align*}
This implies that the bound of Theorem~\ref{thm:general:supp_bound} is
asymptotically optimal for sufficiently large values of $d$, as $\lim_{d\to
  \infty}\log(2\sqrt[d]{6e})=\log(2)=1$.

\subsection{Bounding and Enumerating the Integer Vertices}
By adapting the techniques from Section~\ref{sec:kp:number} and
Section~\ref{sec:kp:enumerating}, we can now apply the support bound of
Theorem~\ref{thm:general:supp_bound} in a generalized setting.
As above, the proximity result of Eisenbrand and
Weismantel~\cite{DBLP:journals/talg/EisenbrandW20} implies that for each vertex
$v\in \Ptop_{I}$, there is an optimal basic feasible solution $y$ of the relaxed
program such that
  $\lVert v-y \rVert_{1}\leq m(2m\Delta+1)^{m}$. 
The box approach of Hayes and Larman~\cite{hayes1983vertices} can be simply
transferred to the $m$-dimensional setting (see
e.\,g.~\cite{schrijver1998theory,DBLP:journals/combinatorica/CookHKM92}). 
Using Theorem~\ref{thm:general:supp_bound}, we can discard all boxes with support
larger than $2m\log(O(\sqrt{m}\Delta))$. 
Hence, for each optimal basic feasible solution $y$, we have at most
\begin{align*}
  2m\log(O(\sqrt{m}\Delta))\cdot n^{2m\log(O(\sqrt{m}\Delta))}\cdot (\log([m(2m\Delta+1)]^{m}))^{2m\log(O(\sqrt{m}\Delta))}
\end{align*}
boxes that each might contain at most one vertex.
As the number of optimal basic feasible solutions $y$ is at most $\binom{n}{m}\leq n^{m}$, the
number of vertices of $\Ptop_{I}$ is at most
\begin{align*}
  n^{m}\cdot 2m\log(O(\sqrt{m}\Delta))\cdot n^{2m\log(O(\sqrt{m}\Delta))}\cdot (\log([m(2m\Delta+1)]^{m}))^{2m\log(O(\sqrt{m}\Delta))}.
\end{align*}
\begin{theorem}
  \label{thm:general:vertices}
  The number of vertices of the integer hull $\Ptop_{I}$ of a general polytope $\Ptop$ is
  at most
  $(n\cdot m \cdot \log(m\Delta))^{O(m\log(\sqrt{m}\Delta))}$. 
\end{theorem}
Note that the best known bound before was $(n\cdot m\cdot
\Delta)^{O(m^{2}\log(\sqrt{m}\Delta))}$ due to Aliev et
al.~\cite{DBLP:journals/siamjo/AlievLEOW18}. 

Using the exact same algorithm described in Section~\ref{sec:kp:enumerating}, we
can also enumerate the vertices in the same running time.
\begin{theorem}
  All integer vertices of the integer hull $\Ptop_{I}$ can be enumerated in time
  $(n\cdot m \cdot \log(m\Delta))^{O(m\log(\sqrt{m}\Delta))}$. 
\end{theorem}

\section{Concentration of Measure for vectors}
\label{sec:measure}
In the following, let $Y_{1},\ldots,Y_{n}$ be independent random variables that
each output a $m$-dimensional vector with $\Pr[Y_{i}[j]\in
[\, \ell_{i}[j],u_{i}[j]\, ]]=1$, i.\,e., the $j$th entry in the vector produced
by the $i$th random variable is at least $\ell_{i}[j]$ and at most
$u_{i}[j]$. 
Let $Y=\sum_{i}Y_{i}$ be the sum of these random variables and $\mu := \Exp[Y]$
be its expected value. Ideally, one wants to show that the maximal derivation
$\lVert Y-\mu \rVert_{\infty}$ is also bounded, but for growing $m$, this
probability shrinks very fast. We thus concentrate on the total sum of the derivations
$\lVert Y-\mu \rVert_{1}$. 
We  show that the random variable $\lVert Y-\mu \rVert_{1}$
has the  concentration of measure effect with regard to the number of variables
$d$ which follows from the following
theorem.

\hoeffdingvector*

Note the difference between the summation orders here. If
$\ell_{i}[j]-u_{i}[j]\leq \Delta$, we can simplify this bound to
\begin{align*}
&\Pr[\lVert Y-\mu \rVert_{1} \geq   1{.}12\cdot  m\cdot \sqrt{n}\cdot \Delta) + \delta]\leq 2\exp(-(2\delta^2)/(nm^{2}\Delta^{2})).
\end{align*}

To prove this theorem, we need a variation on the Azuma-Hoeffding inequality,
often called McDiarmid's inequality. 
\begin{theorem}[Theorem 13.7 in \cite{mitzenmacher2017probability}]
  \label{thm:mcdiarmid}
Let $X_1,X_2,\ldots,X_n$ be independent random variables with range $\Omega$ and
$f$ be any function such that for each $i=1,\ldots, n$, there is some value
$c_{i}$ with 
\begin{align*}
|f(x_1,\ldots,x_i,\ldots,x_n)-f(x_1,\ldots,x_{i-1},\hat{x_i},x_{i+1},\ldots,x_n)|\leq c_{i}
\end{align*}
for all $x_1,\ldots,x_n,\hat{x_i}\in \Omega$. Then
\begin{align*}
\Pr[|f(X_1,\ldots,X_n)-\Exp[f(X_1,\ldots,X_n)]|\geq \delta]\leq 2\exp(-(2\delta^2)/\sum_{i=1}^{n}c_{i}^{2}). 
\end{align*}
\end{theorem}

\begin{proof}[Proof of Theorem~\ref{thm:hoeffding:vector}]
We prove this theorem in two steps: First, we analyze the expected value $\Exp[\lVert Y-\mu \rVert_{1}]$ of $Y$ and second,
we show that $\lVert Y \rVert_{1}$ does not deviate much from its expected
value. Combining these two statements then gives the desired bound.

\subparagraph*{Bounding $\Exp[\lVert Y-\mu \rVert_{1}]$:}
We will first show that a single coordinate $j=1,\ldots,m$ does not deviate much
from its expected value.
Let $Y_i[j]$, $Y[j]$, and $\mu[j]$ be the corresponding projections on
the $j$-th coordinate of $Y_{i}$ (resp.~$Y$ or $\mu$). 
Clearly, we have $\Exp[Y[j]]=\sum_{i}\Exp[Y_{i}[j]]$, but we are
interested in the term $\Exp[\lVert Y[j]\rVert_{1}]=\Exp[|Y[j]|]$ which
might be much more complicated.
We will again use the \emph{Hoeffding bound} of Theorem~\ref{thm:hoeffding} to
bound this term. 

As the $Y_{i}[j]$ are independent for $i=1,\ldots,n$, we conclude
\begin{align*}
  \Pr[|Y[j] - \mu[j]| \geq \delta] \leq 2\exp(-(2\delta^{2})/\sum_{i=1}^{n}(u_{i}[j]-\ell_{i}[j])^{2}.
\end{align*}

The random variable $|Y[j]|$ thus has strong tail bounds which we will use in
the following to bound $\Exp[|Y[j]-\mu[j]|]$:
\begin{restatable}{claim}{boundedexp}
Let $X$ be a non-negative random variable such that for all $\delta > 0$, we have
\begin{align*}
\tag{$\ast$}
\Pr[X\geq \delta]\leq 2\exp(-(2\delta^{2})/b)
\end{align*}
for some $b > 0$. 
Then
\begin{align*}
\Exp[X] \leq 1{.}12\sqrt{b}.
\end{align*}
\end{restatable}

Due to a lack of space, the proof of this claim can be found in the appendix in
Section~\ref{sec:proofs}. 

As $\Pr[|Y[j] - \mu[j]| \geq \delta] \leq 2\exp(-(2\delta^{2})/\sum_{i=1}^{n}(\ell_{i}[j]-u_{i}[j])^{2})$, we have
$\Exp[|Y[j]-\mu[j]|]\leq 1{.}12\cdot \sqrt{\sum_{i=1}^{n}(\ell_{i}[j]-u_{i}[j])^{2}}$. 
Using the linearity of the expected value, we can bound the $1$-norm of the
complete vector $Y -\mu$: 
\begin{align*}
\Exp[\lVert Y -\mu \rVert_{1}] = \Exp[\sum_j |Y[j] - \mu[j]|]=\sum_j \Exp[|Y[j]-\mu[j]|]\leq 1{.}12\cdot  \sum_{j=1}^{m}\sqrt{\sum_{i=1}^{n}(\ell_{i}[j]-u_{i}[j])^{2}}. 
\end{align*}

\subparagraph*{Deviation from the expected value:}
We will now use Theorem~\ref{thm:mcdiarmid}, McDiarmid's inequality, to show that $\lVert Y \rVert_{1}$ will not likely deviate from its expected value.

Using the function $f(Y_1,\ldots,Y_n)=\lVert \sum_{i} Y_{i}
-\mu\rVert_{1}=\lVert Y -\mu \rVert_{1}$, we have
\begin{align*}
|f(x_1,\ldots,x_i,\ldots,x_n)-f(x_1,\ldots,x_{i-1},\hat{x_i},x_{i+1},\ldots,x_n)|\leq \sum_{j=1}^{m} (u_{i}[j]-\ell_{i}[j]). 
\end{align*}

We can thus conclude that
\begin{align*}
  \Pr[|\lVert Y -\mu \rVert -\Exp[\lVert Y -\mu \rVert]|\geq \delta]\leq 2\exp(-(2\delta^{2})/\sum_{i=1}^{n}[\sum_{j=1}^{m} (u_{i}[j]-\ell_{i}[j])]^{2}). 
\end{align*}

Furthermore, we have $\Exp[f(Y_1,\ldots,Y_n)]=\Exp[\lVert Y -\mu \rVert_{1}]\leq
1{.}12\cdot  \sum_{j=1}^{m}\sqrt{\sum_{i=1}^{n}(\ell_{i}[j]-u_{i}[j])^{2}}$
and thus
\begin{align*}
&\Pr\left[\lVert Y-\mu \rVert_{1} \geq   1{.}12\cdot  \sum_{j=1}^{m}\sqrt{\sum_{i=1}^{n}(\ell_{i}[j]-u_{i}[j])^{2}} + \delta\right]\leq\\ &2\exp(-(2\delta^2)/\sum_{i=1}^{n}[\sum_{j=1}^{m} (u_{i}[j]-\ell_{i}[j])]^{2}).\qedhere
\end{align*}

\end{proof}

\bibliographystyle{abbrv}
\bibliography{lib}

\newpage

\section*{Appendix}
\appendix
\section{Missing proofs}
\label{sec:proofs}

\boundedexp*

\begin{proof}
As $X$ is non-negative, we have for every $a\geq 0$ that
\begin{align*}
  \Exp[X] &= \int_{0}^{\infty}\Pr[X\geq t]\mathop{\mathrm{d}t}  = \int_{0}^{a}\Pr[X\geq t]\mathop{\mathrm{d}t} + \int_{a}^{\infty}\Pr[X\geq t]\mathop{\mathrm{d}t}\\
  &\leq a+\int_{a}^{\infty}\Pr[X\geq t]\mathop{\mathrm{d}t}. 
\end{align*}
Assumption ($\ast$) gives us that $\Pr[X\geq t]\leq 2\exp(-(2t^{2})/b)$. Hence
\begin{align*}
a+\int_{a}^{\infty}\Pr[X\geq t]\mathop{\mathrm{d}t}\leq
a+\int_{a}^{\infty}2\exp(-(2t^{2})/b)\mathop{\mathrm{d}t}. 
\end{align*}
As $\lim_{t\to \infty}2\exp(-(2t^{2})/b) = 0$, we have
\begin{align*}
&a+\int_{a}^{\infty}2\exp(-(2\delta^{2})/b)\mathop{\mathrm{d}t} = 
a+2\frac{b}{2a}\cdot\exp(- (2a^{2})/b) =\\
&a+\frac{b}{a}\cdot\exp(- (2a^{2})/b).
\end{align*}
Choosing $a=\alpha\cdot \sqrt{b}$ thus gives
\begin{align*}
  a+\frac{b}{a}\cdot\exp(- (2a^{2})/b) = \sqrt{b}\cdot (\alpha+\exp(-2\alpha^{2})/\alpha). 
\end{align*}
For $\alpha=0{.}9$, we obtain $\alpha+\exp(-2\alpha^{2})/\alpha \leq 1{.}12$ 
and hence $\Exp[X]\leq 1{.}12\sqrt{b}$. 
\end{proof}

\section{Handling General Polytopes with Structures}
\label{sec:structure}

While many integer programs do have a matrix with bounded coefficients, programs
arising from many applications also have a strong structure in addition. 
This structure is usually not well-captured in the support bounds. 
In the following, we show how such a structure can be used to give better bounds
on the support. 
As before, we want to compare $\hyper_{n}$ and $A\cdot \hyper_{n}$. 
Above, we used the  concentration of measure effect to show that a certain
subset $\hyper'_{n}$ exists such that $\hyper'_{n}$ has about the same size as
$\hyper_{n}$, but $|A\cdot \hyper'_{n}| \ll |A\cdot \hyper_{n}|$.
Now, we want to use structural information about the matrix $A$ to show that
$|A\cdot \hyper_{n}|$ cannot be arbitrarily large.

For instance, suppose that the matrix $A$ does contain two identical rows
$a^{(i)}$ and $a^{(j)}$. 
Clearly, we can simply remove $a^{(j)}$ from $A$ and thus reduce the support
bound.
Now, consider the situation, where $a^{(j)}\neq a^{(i)}$, but they are close in
some sense. 
Intuitively, $a^{(i)}x$ and $a^{(j)}x$ should then also be close and thus reduce
the number possible values of $A\cdot \hyper_{n}$.

Fix some matrix $A$ with rows $a^{(1)},\ldots,a^{(m)}$ and consider $B=A\cdot
\hyper_{n}$. For a vector $\beta\in \mathbb{Z}^{m'}$ with $m' < m$, let $S_{A,\beta}$
be the set of different numbers that occur at position $m'+1$ in some element of
$B$, i.\,e.
\begin{align*}
  S_{A,\beta}:=\{b[m'+1]\ \colon\ b\in B, b[i]=\beta[i]\ \forall i=1,\ldots,m'\}. 
\end{align*}
Furthermore, let $n_{A,m'+1}=\max_{\beta\in \mathbb{Z}^{m'}}\{|S_{A,\beta}|\}$ be
the size of a largest such set. 
Furthermore, define $n_{A,1} := |\{b[1]\ \colon\ b\in B\}$ as the possible
values occurring on the first position of an element in~$B$.
For each $b\in B$, there is an $x\in \hyper_{n}$ with $a^{(1)}x=b[1]$. 
The total number of different values of $b[1]$ is thus at most $\lVert a^{(1)}
\rVert_{1}+1$, as $x\in \hyper_{n}$. Hence, $n_{A,1}\leq \lVert a^{(1)}
\rVert_{1}+1$.
For the remaining values $n_{A,m'+1}$, we define for $k=2,\ldots,m$ and for a
subset $\mathbb{X}\subseteq \mathbb{R}$ the value 
$\operatorname{dist}_{1,\mathbb{X}}(a^{(k)};(a^{(1)},\ldots,a^{(k-1)}))$
as the $\ell_{1}$-distance between $a^{(k)}$ and
the closes vector that can be constructed by a linear combination of
$a^{(1)},\ldots,a^{(k-1)}$ with coefficients from $\mathbb{X}$, i.\,e.
\begin{align*}
 \operatorname{dist}_{1,\mathbb{X}}(a^{(k)};(a^{(1)},\ldots,a^{(k-1)}))=\min_{\lambda_{1},\ldots,\lambda_{k-1}\in \mathbb{X}}\{\lVert a^{(k)}-\sum_{i=1}^{k-1}\lambda_{i}a^{(i)} \rVert_{1}\}. 
\end{align*}
We will now see that $n_{A,m'+1}$ can be bounded by this distance for any subset
$\mathbb{X}\subseteq \mathbb{R}$. 

\begin{lemma}
  \label{lem:structure:general_bound}
  For all $m'=1,\ldots,m-1$, we have
  \begin{align*}
  n_{A,m'+1}\leq \lceil
  \operatorname{dist}_{1,\mathbb{X}}(a^{(m'+1)};(a^{(1)},\ldots,a^{(m')}))
  \rceil+1.  
  \end{align*}
  
\end{lemma}

\begin{proof}
  Fix some $m'$ and let $\beta\in \mathbb{Z}^{m'}$ with
  $n_{A,m'+1}=|S_{A,\beta}|$. Hence, there is a non-empty subset
  $\hyper'_{n}\subseteq \hyper_{n}$ with $a^{(i)}x=\beta[i]$ for all
  $i=1,\ldots,m'$ and all $x\in \hyper'_{n}$.
    Let $\lambda_{1}\ldots,\lambda_{m'}\in \mathbb{X}$ be coefficients that attain the
  minimal $\ell_{1}$-distance, i.\,e.
\begin{align*}
  \lVert \sum_{i=1}^{m'}\lambda_{i} a^{(i)}  -a^{(m'+1)} \rVert_{1} = \operatorname{dist}_{1,\mathbb{X}}(a^{(m'+1)},;(a^{(1)},\ldots,a^{(m')})). 
\end{align*}
Now, for all $x\in \hyper'_{n}$, we have
$\sum_{i=1}^{m'}\lambda_{i}a^{(i)}x  =
\sum_{i=1}^{m'}\lambda_{i}\beta[i]$. Denote this value by $v$.
For all $s\in S_{A,\beta}$, we thus have
$(\sum_{i=1}^{m'}\lambda_{i}a^{(i)}-a^{m'+1})x=v-a^{m'+1}x=v-s$ for some $x\in
\hyper'_{n}$. As $x\in \hyper_{n}$, we can conclude that
\begin{align*}
  |v-s|\leq \lVert \sum_{i=1}^{m'}\lambda_{i}a^{(i)}-a^{m'+1} \rVert_{1}. 
\end{align*}
Hence, $s$ can take at most
\begin{align*}
\sum_{i=1}^{m'}\lambda_{i}a^{(i)}-a^{m'+1} \rVert_{1} \rceil+1 = \lceil
\operatorname{dist}_{1,\mathbb{X}}(a^{(m'+1)};(a^{(1)},\ldots,a^{(m')}))
\rceil+1
\end{align*}
 possible values and thus $n_{A,m'+1} = |S_{A,\beta}|\leq \lceil
\operatorname{dist}_{1,\mathbb{X}}(a^{(m'+1)};(a^{(1)},\ldots,a^{(m')}))
\rceil+1$. 
\end{proof}

Furthermore, note that we can insert the rows in an arbitrary order and that
$\mathbb{X}=\mathbb{R}$ gives us the smallest distances. 
As $|A\cdot \hyper_{n}|\leq \prod_{i=1}^{m}n_{A,i}$, we obtain the following useful lemma.

\begin{lemma}
  \label{lem:structure:hyper_bound}
  Let $a^{(1)},\ldots,a^{(m)}$ be the rows of an integral matrix $A$.
  For a permutation $\pi$ on $\{1,\ldots,m\}$, 
  define $d^{(\pi)}_{1}=\lVert a^{(\pi(1))} \rVert$ and
$d^{(\pi)}_{i+1}=\lceil  \operatorname{dist}_{1,\mathbb{R}}(a^{(\pi(i+1))};(a^{(\pi(1))},\ldots,a^{(\pi(i))}))\rceil$
  for $i=1,\ldots,m-1$. Then, $|A\cdot \hyper_{n}|\leq \min_{\pi}\{\prod_{i=1}^{m}
  (d^{(\pi)}_{i}+1)\}$. 
\end{lemma}

Now, consider a solution $x^{*}$ of $\{x\in \ZZgeq^{n}\mid Ax=b\}$ with support
$\supp(x^{*})=S$ of size $s=|S|$. 
If $\prod_{i=1}^{m}(d^{(\pi)}_{i}+1) < 2^{s}$, there are two points $x,x'\in
\hyper_{s}$ with $A[S]x=A[S]x'$.

Using Lemma~\ref{lem:eisenbrand:shmonin}, we know that $x^{*}$ is not a vertex.
We thus obtain the following bound. 
\begin{lemma}
  \label{lem:structure:supp_bound}
  For each vertex $v$ of the integer hull $\Ptop_{I}$ we have that 
  \begin{align*}
      |\supp(v)| \leq \min_{\pi}\{\sum_{i=1}^{m}\log(d^{(\pi)}_{i}+1)\}. 
  \end{align*}
\end{lemma}

To give a simple example where such a bound can be useful, consider the matrix
\begin{align*}
  A=
  \begin{pmatrix}
    1 & 0 & 0 & \ldots & 0\\
    1 & 1 & 0 & \ldots & 0\\
    1 & 1 & 1 & \ldots & 0\\
    \vdots & \vdots & \vdots &\ldots & \vdots
  \end{pmatrix}
\end{align*}
having $m$ rows. Using Theorem~\ref{thm:general:supp_bound} gives a support bound
of about $m\log(m)$, while Lemma~\ref{lem:structure:supp_bound} directly gives a
better bound of $m$, as $d_i=1$.

\paragraph*{Using Minkowski's Second Theorem}
To obtain a bound that is easier to handle, we can make use of Minkowski's
second theorem. For a set of $m\leq $ linear independent vectors
$B=\{B_{1},\ldots,B_{n}\}\subseteq \mathbb{R}^{n}$,
the \emph{lattice} $\Lambda(B)$  of rank $m$ is defined as
\begin{align*}
  \Lambda(B)=\{\sum_{i=1}^{m}\alpha_{i}B_{i}\mid \alpha_{1},\ldots,\alpha_{n}\in \mathbb{Z}\}. 
\end{align*}
Let $C\subseteq \mathbb{R}^{n}$ be a central symmetric convex body.
For $i=1,\ldots,m$, the $i$th \emph{successive minimum} (with regard to
$\Lambda(B)$ and $C$) $\lambda_{i}$ is defined as the smallest
positive real $\lambda$ such that $\lambda C$ contains at least $i$ linearly
independent points of $\Lambda(B)$. Alternatively, $\lambda_{1}$ is the length
of the shortest non-zero vector in $\Lambda(B)$ and for $i > 1$, the value
$\lambda_{i}$ is the length of the shortest vector linear independent of the
vectors corresponding to $\lambda_{1},\ldots,\lambda_{i}$.

\begin{theorem}[Minkowski's second theorem \cite
  {gruber1987geometry} (Chapter 2, Paragraph 9.1, Theorem 1)]
  Let $\Lambda(B)\subseteq \mathbb{R}^{n}$ be a lattice of rank $m$ and $C\subseteq
  \mathbb{R}^{n}$ be a central symmetric convex body. Furthermore, let
  $\lambda_{1},\ldots,\lambda_{m}$ be the successive minima with regard to
  $\Lambda(B)$ and $C$. Then
  \begin{align*}
    \prod_{i=1}^{m}\lambda_{i}\leq 2^{m}\cdot \frac{\det(\Lambda(B))}{\operatorname{vol}(C\cap \operatorname{span}(\Lambda(B)))}. 
  \end{align*}
\end{theorem}

Now, consider the lattice $\Lambda(a^{(1)},\ldots,a^{(m)})$ constructed by the
rows of $A$. By definition, we have 
$\det(\Lambda(a^{(1)},\ldots,a^{(m)}))=\sqrt{\det(AA^{\top})}$.
As we want to concentrate on the $\ell_{1}$-norm, we consider the 
cross-polytope $C^{\times}_{n}=\{x\in \mathbb{R}^{n}\colon \lVert
x \rVert_{1}\leq 1\}$ that is a central symmetric convex body.
As the volume of $C^{\times}_{n}$ is $2^{n}/n!$~\cite{DBLP:reference/cg/HenkRZ04},
and $\Lambda(a^{(1)},\ldots,a^{(m)})$ has rank $m$, we have 
$\operatorname{vol}(C^{\times}_{n}\cap \operatorname{span}(\Lambda(a^{(1)},\ldots,a^{(m)})))\leq \operatorname{vol}(C^{\times}_{m})=
2^{m}/m!$.

Let $v_{1},\ldots,v_{m}$ be the shortest linear independent vectors
corresponding to the successive minima $\lambda_{1},\ldots,\lambda_{m}$ (with
regard to $\Lambda(a^{(1)},\ldots,a^{(m)})$ and $C^{\times}_{n}$). As these
vectors are elements of $\Lambda(a^{(1)},\ldots,a^{(m)})$, we know
$\Lambda(a^{(1)},\ldots,a^{(m)})=\Lambda(v_{1},\ldots,v_{m})$. Let $V$ be the
matrix with rows $v_{1},\ldots,v_{m}$. By Lemma~\ref{lem:structure:hyper_bound},
we know that $|V\cdot \hyper_{n}|\leq \prod_{i=1}^{m}(d_{i}+1)$, where
$d_{1}=\lVert v_{1} \rVert_{1}=\lambda_{1}$ and $d_{i}=\lceil
\operatorname{dist}_{1,\mathbb{R}}(v_{i}; (v_{1},\ldots,v_{i-1}))\rceil\leq
\lVert v_{i} \rVert_{1}=\lambda_{i}$. The last inequality follows from the fact
that we can set all coefficients of $v_{1},\ldots,v_{i-1}$ to $0$. Hence,
$|V\cdot \hyper_{n}|\leq \prod_{i=1}^{m}(\lambda_{i}+1)$. As
$\Lambda(a^{(1)},\ldots,a^{(m)})$ is an integral lattice, we know that all
$v_{i}$ are integral and hence, all $\lambda_{i}$ are integral also. Hence,
$\lambda_{i}+1\leq 2\lambda_{i}$. 

 Minkowski's second theorem
now states that
\begin{align*}
  \prod_{i=1}^{m} \lambda_{i} \leq 2^{m}\cdot \sqrt{\det(AA^{\top})}/(2^{m}/m!) =m!\cdot \sqrt{\det(AA^{\top})}. 
\end{align*}
Hence 
\begin{align*}
  |V\cdot \hyper_{n}|\leq \prod_{i=1}^{m}2\lambda_{i}\leq 
 2^{m}\cdot m!\cdot \sqrt{\det(AA^{\top})}. 
\end{align*}

Now, consider a solution $x^{*}$ of $\{x\in \ZZgeq^{n}\mid Ax=b\}$ with support
$\supp(x^{*})=S$ of size $s=|S|$. 
If $2^{m}\cdot m!\cdot \sqrt{\det(AA^{\top})} < 2^{s}$, there are thus two points $x,x'\in
\hyper_{s}$ with $V[S]x=V[S]x'$. As $v_{i}\in \Lambda(a^{(1)},\ldots,a^{(m)})$,
this implies that $A[S]x=A[S]x'$. 
Using Lemma~\ref{lem:eisenbrand:shmonin}, we know that $x^{*}$ is not a vertex.
We thus obtain the following bound. 
\begin{lemma}
  \label{lem:structure:supp_bound:minkowski}
  For each vertex $v$ of the integer hull $\Ptop_{I}$ we have that 
  \begin{align*}
      |\supp(v)| \leq m+m\log(m)+\log(\sqrt{\det(AA^{\top})}). 
  \end{align*}
\end{lemma}

Note that this nearly matches the result of Aliev et al.~\cite{aliev2017sparse}
that obtain a bound of $m+\log(\sqrt{\det(AA^{\top)}})$.

\section{Analysing the Inequalities}
\label{sec:inequality}
In this section we study inequalities of the form
\begin{align}
   |Y| - m/2 \log(|Y|) > m \log(c \Delta) \label{log:inequality}
\end{align}
for $c \ge 2$ depending on the relation of the parameters $c$, $\Delta$, where $Y$ is some finite set. 
Here $m$ is the number of constraints (rows) of $A$ and $\Delta$ is the largest
absolute value of a coefficient in $A$. 
We suppose that all entries in $A$ are integral. 
The main goal is to estimate the smallest cardinality $|Y|$ such that the
inequality above holds. 
As shown above, this implies an upper bound for the support of any optimum ILP solution with
minimum number of positive entries $x_i > 0$. 
Using $\Delta \ge 1$ and $c \ge 2$, we notice that $|Y| \ge 2$.

\subsection{A Warmup}
\label{sec:inequality:warmup}
To get a feeling for the kind of arguments that we will use, we first 
use a similar approach to Eisenbrand and
Shmonin~\cite{eisenbrand2006caratheodory} to simplify the inequality
\begin{align*}
  |Y| \leq m\log(4{.}24\cdot e\sqrt{|Y|}\Delta+2e)
\end{align*}
derived in Theorem~\ref{thm:general:supp_bound}. 

\begin{lemma}
  \label{lem:ineq:supp_bound}
  For each $\epsilon > 0$, the smallest cardinality $|Y|$ to fulfill $|Y| \leq
  m\log(4{.}24\cdot e\sqrt{|Y|}\Delta+2e)$ can be bounded by 
   \begin{align*} 
     |Y| \leq (1+\epsilon)m\cdot \log(4{.}24\cdot e\cdot  (1+\Delta^{-1})\cdot \sqrt{(1+\epsilon)/(2\epsilon)}\cdot \sqrt{m}\cdot \Delta).
   \end{align*} 
 \end{lemma} 
 \begin{proof} 
   Assume that 
   \begin{align*} 
     |Y| > (1+\epsilon)m\cdot \log(4{.}24\cdot e\cdot  (1+\Delta^{-1})\cdot \sqrt{(1+\epsilon)/(2\epsilon)}\sqrt{m}\Delta)
   \end{align*} 
   We then have 
   \begin{align*} 
     &|Y| > (1+\epsilon)m\cdot \log(4{.}24\cdot e\cdot (1+\Delta^{-1})\cdot \sqrt{(1+\epsilon)/(2\epsilon)}\sqrt{m}\Delta) \Leftrightarrow_{x\to 2^{x}}\\ 
     &2^{|Y|} > 2^{(1+\epsilon)m\cdot \log(4{.}24\cdot e\cdot (1+\Delta^{-1}) \sqrt{(1+\epsilon)/(2\epsilon)}\sqrt{m}\Delta)}\Leftrightarrow\\ 
     &2^{|Y|} > (4{.}24\cdot e\cdot (1+\Delta^{-1}) \sqrt{(1+\epsilon)/(2\epsilon)}\sqrt{m}\Delta)^{(1+\epsilon)m}\Leftrightarrow_{x \to \sqrt[(1+\epsilon)m]{x}}\\ 
     & 2^{|Y|/((1+\epsilon)m)} > 4{.}24\cdot e\cdot (1+\Delta^{-1}) \sqrt{(1+\epsilon)/(2\epsilon)}\sqrt{m}\Delta\Leftrightarrow\\ 
     & 2^{|Y|/((1+\epsilon)m)}/(4{.}24\cdot e\cdot (1+\Delta^{-1}) \sqrt{(1+\epsilon)/(2\epsilon)}\sqrt{m}) > \Delta.\tag{$\ast$} 
   \end{align*} 
   This implies 
   \begin{align*} 
     &m\cdot \log(4{.}24\cdot e\cdot \sqrt{|Y|}\cdot \Delta+2e)= \\
     &m\cdot \log(2e ( 2{.}12\cdot \sqrt{|Y|}\cdot \Delta+1))\leq \\
     &m\cdot \log(2e ( 2{.}12\cdot \sqrt{|Y|}\cdot (1+\Delta^{-1})\Delta))=\\
     &m\cdot \log(4{.}24\cdot e\cdot \sqrt{|Y|}\cdot (1+\Delta^{-1})\Delta) < \\
     &m\cdot \log(4{.}24\cdot e\cdot \sqrt{|Y|}\cdot (1+\Delta^{-1}) 2^{|Y|/((1+\epsilon)m)}/\\
     &\quad\quad (4{.}24\cdot e\cdot (1+\Delta^{-1}) \sqrt{(1+\epsilon)/(2\epsilon)}\sqrt{m})) =\\
     &m\cdot \log(2^{|Y|/((1+\epsilon)m)})+m\cdot \log(4{.}24\cdot e\cdot \sqrt{|Y|}\cdot (1+\Delta^{-1})/\\
     &\quad\quad(4{.}24\cdot e\cdot (1+\Delta^{-1}) \sqrt{(1+\epsilon)/(2\epsilon)}\sqrt{m}))=\\
     &m\cdot \log(2^{|Y|/((1+\epsilon)m)})+m\cdot \log( \sqrt{|Y|}/( \sqrt{(1+\epsilon)/(2\epsilon)}\sqrt{m}))=\\
     &m\cdot \log(2^{|Y|/((1+\epsilon)m)})+(m/2)\cdot \log( |Y|/( (1+\epsilon)/(2\epsilon)m))=\\
     &m\cdot |Y|/((1+\epsilon)m)+(m/2)\cdot \log( |Y|/( (1+\epsilon)/(2\epsilon)m))\leq \\
     &m\cdot |Y|/((1+\epsilon)m)+(m/2)\cdot  |Y|/( (1+\epsilon)/(2\epsilon)m)=\\
     &|Y|/(1+\epsilon) + \epsilon |Y|/(1+\epsilon) = |Y|.
     \end{align*} 
   This in turn is a contradiction to Theorem~\ref{thm:general:supp_bound}. 
 \end{proof}

\subsection{The general Case}
Suppose that $|Y| = m \log(c \Delta) + y$. Then
$|Y| =   m \log(c \Delta) + y > m \log(c \Delta) + (m/2) \log (|Y|) \ge m \log(c \Delta) + m/2$. This implies that $y > m/2 \ge 0$. Now let $y = (m/2) \log(\alpha)$. This is equivalent to $\alpha = 2^{2y/m}$ and implies
that $\log(\alpha) > 1$ or $\alpha > 2$. Now we need that
$$ y = (m/2) \log(\alpha) > (m/2) \log(m \log(c \Delta) + y).$$
This inequality holds when the arguments of the log function are also
larger:
$$ \begin{array}{lcl} \alpha & > & m \log(c \Delta) + (m/2) \log(\alpha) \\
                             & = & (m/2) \log(c \Delta)^2 + (m/2) \log(\alpha) \\
                             & = & (m/2) \log[(c\Delta)^2 \alpha]
                             \end{array} $$
Therefore, we obtain $\alpha >   (m/2) \log[(c\Delta)^2 \alpha]$ and
$\alpha > m \log(c \Delta)$. Now we set $\alpha = m \log(c \Delta) + \beta$
where $\beta > 0$. Using this setting we get
$$ \alpha = m \log(c \Delta) + \beta > m \log(c \Delta) + (m/2) \log[m \log(c \Delta) + \beta].$$
Therefore, $\beta > (m/2) \log[m \log(c \Delta) + \beta]$. Here we use now
$\beta = (m/2) \log[m \log(c \Delta) + \gamma]$ where $\gamma > \beta > 0$.
This implies
$$ (m/2) \log[m \log(c\Delta) + \gamma] > (m/2) \log [m \log(c \Delta) + (m/2) \log[m \log(c \Delta) + \gamma]].$$
This is equivalent to
$$ m \log(c \Delta) + \gamma > m \log(c \Delta) + (m/2) \log[m \log(c \Delta) + \gamma].$$
and
$$ \gamma > (m/2) \log[m \log(c \Delta) + \gamma].$$

\begin{lemma}
Suppose that we have a concrete value for $\gamma$ such that $ \gamma > (m/2) \log[m \log(c \Delta) + \gamma].$ Then, using the definition 
$\beta = (m/2) \log(m/2) \log[m \log(c \Delta) + \gamma]$, $\alpha = m \log(c \Delta) + \beta$, $y = m/2 \log(\alpha)$ and $|Y| = m \log(c \Delta) + y$
we obtain $|Y| - m/2 \log(|Y|) > m \log(c \Delta)$.
\end{lemma}

\subsubsection{Case 1} 

Here we consider the case that $m < (2/3) \frac{c \Delta}{\log(c \Delta)}$.
If we use $\gamma = (m/2) \log(c \Delta)$ we obtain the inequality
$$ \gamma = (m/2) \log(c \Delta) > (m/2) \log[m \log(c \Delta) + (m/2) \log(c \Delta)]$$
or equivalent
$$ c \Delta > m \log(c \Delta) + (m/2) \log(c \Delta) = (3m/2) \log(c \Delta).$$
Since $m < 2/3 {c \Delta \over \log(c \Delta)}$, the above inequality holds directly. In this case we obtain for
$$ \begin{array}{lcl} |Y| & = & m \log(c \Delta) + (m/2) \log(\alpha) \\
                 & = & m \log(c \Delta) + (m/2) \log[m \log(c \Delta) + \beta] \\
                 & = & m \log(c \Delta) + (m/2) \log[m \log(c \Delta) +
                 (m/2) \log((3m/2) \log(c \Delta))] \\
                 & \le & m \log(c \Delta) + (m/2) \log[m \log(c \Delta) +
                 (m/2) \log(c \Delta)] \\
                 & = & m \log(c \Delta) + (m/2) \log[(3m/2) \log(c\Delta)] \\
                 & = & m \log(c \Delta) + m \log[(3m/2) \log(c \Delta)]^{1/2} \\
                 & = & m \log(c \Delta \sqrt{(3m/2) \log(c \Delta)})
      \end{array} $$
using $ (3 m/2) \log(c \Delta) \le c \Delta$. Notice that the term above can be
bounded also by $m \log(c \Delta)^{3/2} = (3/2 m) \log(c \Delta)$. Alternatively, we can obtain the following bound
$$ |Y|  \le m \log(c \Delta \sqrt{m \log(c \Delta \sqrt{(3m/2) \log(c \Delta)})}).$$
Using $(3m/2) \log(c\Delta) \le c \Delta$, we also get the same bound
$|Y| \le m \log(c \Delta \sqrt{m \log(c \Delta)^{3/2}}) = m \log(c \Delta \sqrt{(3m/2) \log(c \Delta)})$.

\subsubsection{Case 2} Here we consider the case that $m \ge 2/3 {c \Delta \over \log(c \Delta)}$ or equivalently $c \Delta \le (3m/2) \log(c \Delta)$. In addition suppose that $\log(c \Delta) < (c \Delta)^{1/2}$. Notice that
$(c \Delta)^{1/2} (c \Delta)^{1/2} = (c \Delta) \le (3m/2) \log(c \Delta) <
(3 m/2) (c \Delta)^{1/2}$. This implies that $(c \Delta)^{1/2} < 3m/2$ or
$(c \Delta) < (3m/2)^2$.
Therefore, $\log(c \Delta)^{1/2} < \log(3m/2)$. Another consequence is that
$1/2 \log(c \Delta) < \log(3m/2)$ or $\log(c \Delta) < 2 \log(3m/2) = \log(3m/2)^2 < 3m/2$ for $c \Delta \ge 4$. For the last inequality we use the argument in the following lemma about the function
$f(x) = x^{1/2} - \log(x)$. We obtain $\log(3m/2) < 3m/4$ for $c \Delta \ge 4$; otherwise we can bound $\log(3m/2) < 3m/2$.

\begin{lemma} \label{log:bound1} Suppose that $f(x) = x^{1/2} - \log(x)$.
If $f(x) > 0$ for $x \ge 9$ and $y > x $. Then $f(y) > 0$.
\end{lemma}
\begin{proof}
Consider $f'(x) = 1/2 x^{-1/2} - 1/(x \ln(2))$. The function is strong monotone increasing at $x$, if $f'(x) \ge 0$. This means that
$1/2 x^{-1/2} >   1/(x \ln(2))$ or equivalent $x^{1/2} > 2/\ln(2)$
or $x > (2 / \ln(2))^2 \approx 8.3254$. Notice that $f(x) \ge 0$ for $x=4$ and $x = 16$, but $f(x) < 0$ for $x \in (4,16)$. Therefore, $x \ge 9$ implies that $x > 16$. Using the monotonicity, $f(y) > 0$.
\end{proof}

If $\log(c \Delta) < (c \Delta)^{1/2}$ and $c \Delta \ge 4$, then $c \Delta > 16$ and $\log(y) < \sqrt{y}$ for any $y > c \Delta$. In the inequality before the lemma we used $y = (3m/2)^2 > c \Delta$.

Now we calculate the inequality $\gamma > (m/2) \log(m \log(c \Delta) + \gamma)$ for $c \Delta \le (3m/2) \log(c \Delta)$. The right hand side for $c \Delta \ge 4$ is at most
$$ \begin{array}{lcl}
m/2 \log(m \log(3m/2 \log(c \Delta)) + \gamma) &
\le & m/2 \log(m \log(3m \log(3/2 m))+\gamma) \\
  & < &  m/2 \log(m \log(9/4 m^2)+\gamma) \\
  & = &  m/2 \log(2m \log(3/2 m)+ \gamma)
  \end{array} $$

We calculate now the property for $\gamma = m \log(3m/2)$. We obtain the
following condition $m \log(3m/2) \ge m/2 \log(3m \log(3/2 m))$ or
$\log(3/2m) \ge \log(3m \log(3/2 m))^{1/2}$. This inequality holds if
$3/2 m \ge (3m \log(3/2m))^{1/2}$ or $9/4 m^2 = (3/2 m)^2 \ge 3m \log(3/2 m)$
or $3/4 m \ge \log(3/2 m)$.

Without the assumption on $c \Delta$ we obtain
$m/2 \log(2m \log(3/\sqrt{2} m) + \gamma)$. Here we can calculate the
property for $\gamma = m \log(3m/\sqrt{2})$. The condition $m \log(3m/\sqrt{2}) \ge m/2 \log(3m \log(3/\sqrt{2} m))$ or
$\log(3/\sqrt{2} m) \ge \log(3m \log(3/\sqrt{2} m))^{1/2}$. This inequality holds if
$3/\sqrt{2} m \ge (3m \log(3/\sqrt{2} m))^{1/2}$ or $9/2 m^2 = (3/\sqrt{2} m)^2 \ge 3m \log(3/\sqrt{2} m)$
or $3/2 m \ge \log(3/\sqrt{2} m)$. The last inequality holds for all
$m \ge 1$.

\begin{lemma}
We have $\log(x) \le x / \sqrt{2}$ for $x \ge 2.1$.
\end{lemma}
\begin{proof}
To prove this consider $f(x) = x /\sqrt{2} - \log(x)$.
Here we have $f'(x) = 1 /\sqrt{2} - 1/(x \ln(2)) \ge 0$ iff
$x \ge \sqrt{2}/\ln(2) \approx 2.040$. In addition we have $f(2) = 2 /\sqrt{2} - \log(2) > 0.414$.
\end{proof}

Notice that the argument $x$ for the log function is equal to
$3/\sqrt{2} m \ge 2.121 m \ge 2.121$.

In the following we estimate the value for $|Y|$ for $c \Delta \ge 4$. We obtain
$$ \begin{array}{lcl}
    |Y| & = & m \log(c \Delta) + m/2 \log(\alpha) \\
        & = & m \log(c \Delta) + m/2 \log(m \log(c \Delta) + \beta) \\
        & = & m \log(c \Delta) + m/2 \log(m \log(c \Delta) +
             m/2 \log[m \log(c \Delta) + m \log(3/2 m)]) \\
        &  \le &   m \log(c \Delta) + m/2 \log(m \log(c \Delta) +
             m/2 \log[m \log(3m/2)^2 + m \log(3/2 m)]) \\
        & = &     m \log(c \Delta) + m/2 \log(m \log(c \Delta) +
             m/2 \log[3m \log(3m/2)]) \\
        & \le &  m \log(c \Delta) + m/2 \log(2m \log(3/2 m) +
             m/2 \log[3m \log(3m/2)]) \\
        & < &  m \log(c \Delta) + m/2 \log(3m \log(3m/2)) \\
        & = &  m \log(c \Delta \sqrt{3m \log(3m/2)}).
        \end{array}  $$

In the calculation above we used the inequality $m/2 \log(3m \log(3m/2)) =
m \log(3m \log(3/2m))^{1/2} <
m \log(3m/2)$. To see this, notice that $(3 m \log(3/2 m))^{1/2} < 3/2 m$
or $3 m \log(3/2 m) < (3/2 m)^2 = 9/4 m^2$ or $\log(3/2 m) < 3/4 m$.

In the general case (without the assumption von $c \Delta$) we get
$$ \begin{array}{lcl}
    |Y| & = & m \log(c \Delta) + m/2 \log(\alpha) \\
        & = & m \log(c \Delta) + m/2 \log(m \log(c \Delta) + \beta) \\
        & = & m \log(c \Delta) + m/2 \log(m \log(c \Delta) +
             m/2 \log[m \log(c \Delta) + m \log(3/\sqrt{2} m)]) \\
        &  \le &   m \log(c \Delta) + m/2 \log(m \log(c \Delta) +
             m/2 \log[m \log(3m/2)^2 + m \log(3/\sqrt{2} m)]) \\
        & = &     m \log(c \Delta) + m/2 \log(m \log(c \Delta) +
             m/2 \log[3m \log(3m/\sqrt{2})]) \\
        & \le &  m \log(c \Delta) + m/2 \log(2m \log(3/2 m) +
             m/2 \log[3m \log(3m/\sqrt{2})]) \\
        & <  &   m \log(c \Delta) + m/2 \log(2m \log(3/2 m) +
                   m \log(3/\sqrt{2} m)) \\
        & = &    m \log(c \Delta) + m/2 \log(3m \log(3/\sqrt{2} m)) \\
        & = &  m \log(c \Delta \sqrt{3m \log(3m/\sqrt{2})}).
        \end{array}  $$

Here we can estimate $(3m \log(3m /\sqrt{2}))^{1/2} < 3/\sqrt{2} m$. This holds using $3m \log(3m/\sqrt{2}) < 9/2 m^2$ or $\log(3m / \sqrt{2}) < 3/2 m$.

\subsubsection{Case 3} Here we consider the case that $m \ge 2/3 {c \Delta \over \log(c \Delta)}$ or equivalently $c \Delta \le 3m/2 \log(c \Delta)$. In addition suppose that $(c \Delta)^{1/2} \ge 3m/2$. In this case we obtain
$(c \Delta)^{1/2} \ge 3m/2 \ge {c \Delta \over \log(c \Delta)}$. This implies
$(c \Delta)^{1/2} \log(c \Delta) \ge c \Delta$ or equivalently
$\log(c \Delta) \ge (c \Delta)^{1/2}$.

Now we study $\gamma > (m/2) \log(m \log(c \Delta) + \gamma)$ and use $m  \le (2/3) (c \Delta)^{1/2}$. Therefore let us consider just
$ \gamma > (m/2) \log(2/3 (c \Delta)^{1/2} \log(c \Delta) + \gamma)$ und test
$\gamma = m \log( (c \Delta)^{1/2} \log(c \Delta))$. We obtain the following
condition
$$ m \log( (c \Delta)^{1/2} \log(c \Delta)) > (m/2) \log[2/3 (c \Delta)^{1/2} \log(c \Delta) + m \log( (c \Delta)^{1/2} \log(c \Delta))]$$
The right hand side can be bounded by $ (m/2) \log[(2/3) (c \Delta)^{1/2} \log(c \Delta) + 2/3 (c \Delta)^{1/2} \log( (c \Delta)^{1/2} \log(c \Delta))]$.
Using $(c \Delta)^{1/2} \log(c \Delta) \ge c \Delta$, the value above is at most $ (m/2) \log[(4/3) (c \Delta)^{1/2} \log( (c \Delta)^{1/2} \log(c \Delta))]$. Then, the condition for $\gamma$ holds, if we can prove that
$$ (c \Delta)^{1/2} \log(c \Delta)  > ((4/3) (c \Delta)^{1/2} \log( (c \Delta)^{1/2} \log(c \Delta)))^{1/2}.$$
This can be rewritten as
$$ (c \Delta) \log^2(c \Delta)  > (4/3) (c \Delta)^{1/2} \log( (c \Delta)^{1/2} \log(c \Delta)))$$ or equivalently
$$ (3/4) (c \Delta)^{1/2} \log^2(c \Delta)  >  \log( (c \Delta)^{1/2} \log(c \Delta))).$$
The right hand side is at most $< \log((c \Delta)^{1/2} c \Delta) =
(3/2) \log(c \Delta)$. Therefore, the inequality holds if
$(3/4) (c \Delta)^{1/2} \log^2(c \Delta) \ge  (3/2) \log(c \Delta)$ or
$ (c \Delta)^{1/2} \log(c \Delta) \ge 2$. Since $(c \Delta)^{1/2} \log(c \Delta) \ge c \Delta$, $c \Delta \ge 2$ is sufficient. Notice that
$1 \le m \le (2/3) (c \Delta)^{1/2} \le (2/3) \sqrt{2} < 1$ for $c \Delta \le 2$
gives a in this case contradiction (without the assumption $c \Delta \ge 2$ that we made at the beginning).

Now let us calculate the cardinality for $|Y|$. Using
$(c \Delta)^{1/2} \log(c \Delta) \ge c \Delta$,
$\log( (c \Delta)^{1/2} \log(c \Delta))
< (3/2) \log(c \Delta)$ and
$c \Delta \le (3m/2) \log(c \Delta)$

$$ \begin{array}{lcl}
   |Y| & = & m \log(c \Delta) + m/2 \log[m \log(c \Delta) + m/2 \log(m \log(c \Delta) + m \log ((c \Delta)^{1/2} \log(c \Delta)))]  \\
       & \le & m \log(c \Delta) + m/2 \log[m \log(c \Delta) + m/2 \log(2m \log ((c \Delta)^{1/2} \log(c \Delta)))] \\
       & \le & m \log(c \Delta) + m/2 \log[m \log(c \Delta) + m/2 \log(3m \log (c \Delta))] \\
       & \le & m \log(c \Delta) + m/2 \log[m \log(3m/2 \log(c \Delta)) + m/2 \log(3m \log (c \Delta))] \\
       & \le & m \log(c \Delta) + m/2 \log[3m/2 \log(3m \log(c \Delta))] \\
       & =  &  m \log[c \Delta \sqrt{3m/2 \log(3m \log(c \Delta))}]
       \end{array} $$
 Using $m \le 2/3 \sqrt{c \Delta} \le 2/3 \log(c \Delta)$ the right hand side
 is at most
 $$ \begin{array}{ll}
   & m \log(c \Delta \sqrt{  \log(c \Delta) \log(2 \log^2(c \Delta))}) \\
  \le & m \log(c \Delta \sqrt{ \log^2(2 \log^2(c \Delta))} ) \\
  \le & m \log(c \Delta \log(2 \log^2(c \Delta))) \end{array} $$

Notice that case 3 does not occur in many cases.
First we already observed that $f(x) = (x)^{1/2} - \log(x) \le 0$ only for $x \in [4,16]$. This implies that
$c \Delta$ has to be in the interval $[4,16]$. Next we check first in which cases we have an integral number $m \ge 2$  such that $m \le 2/3 \sqrt{c \Delta}$ and
$m \ge (2/3) {c \Delta \over \log(c \Delta)}$. To analyse this we study
when there is an integral $m \ge 2$ such that $2/3 (x/\log x) \le m \le 2/3 \sqrt{x}$. Using $x \le 16$ we get $ m \le (2/3) \sqrt{16} = 8/3$; this shows that here only $m=2$ is possible. In addition we have $(2/3) (x/\log x) \le 2 \le 2/3 \sqrt{x}$ only for $x \ge 9$ (using $2 \le 2/3 \sqrt{x}$) and $x \le 10$ (using that $g(x) = x/\log x$ is monotone increasing for $x \ge 2.719$ and $(2/3)g(9) < 2$ and $(2/3) g(10) > 2$). This implies that for $m=2$ we need in this case the property $c \Delta \in [9,10]$.
The case $m=1$ can be handled in an easier way. For $m=1$ we get here $x \ge 9/4$ (using $1 \le 2/3 \sqrt{x}$) and $x < 4$ (using that $g$ monotone increasing for $x \ge 2.719$ and $g(4) = 2 > 3/2 \ge x / \log x  $ for each feasible $x$ and $m=1$). Since $x \in [4,16]$, this gives a contradiction.

\subsubsection{Case 4} In  the remaining case we have $m \ge 2/3 {c \Delta \over \log(c \Delta)}$ or equivalently $c \Delta \le (3m/2) \log(c \Delta)$, and in addition suppose that $(c \Delta)^{1/2} <  3m/2$ and $\log(c \Delta) \ge (c \Delta)^{1/2}$. This implies that $c \Delta < (3m/2)^2$ and $(c \Delta)^{1/2} \log(c \Delta) \ge c \Delta$ and $(c \Delta)^{1/2} \ge {c \Delta \over \log(c \Delta)}$.

Now we analyse when $\gamma > m/2 \log(m \log(c \Delta)) + \gamma)$. We bound the right hand side as follows:
$$ \begin{array}{ll}
       & m/2 \log(m \log(c \Delta)) + \gamma) \\
  \le  & m/2 \log(m \log((3m/2) \log(c \Delta)) + \gamma) \\
  =    & m/2 \log(m \log(3m \log(c \Delta)^{1/2}) + \gamma) \\
  <    & m/2 \log(m \log(3m \log(3m/2)) + \gamma) \\
  <    & m/2 \log(m \log(9m^2/2) + \gamma) \\
  \le  & m/2 \log(2m \log((3/\sqrt{2}) m) + \gamma)
  \end{array} $$
Now we set $\gamma = m \log(3m)$ and test whether $m \log(3m) \ge m/2 \log[2m \log(3/\sqrt{2} m) + m \log(3m)]$. This is equivalent to
$$ \log(3m) \ge \log[(3m \log(3m))^{1/2}] $$
or $ 3m \ge (3m \log(3m))^{1/2}$, $9 m^2 \ge 3m \log(3m)$ or $3m \ge \log(3m)$ that holds strictly for any $m \ge 1$.

Inserting $\gamma$ into $|Y|$ gives
$$ \begin{array}{lcl}
 |Y| & =   & m \log(c \Delta) + m/2 \log[m \log(c \Delta) + m/2 \log(m \log(c \Delta) + m \log(3m))] \\
     & \le & m \log(c \Delta) + m/2 \log[m \log(c \Delta) + m/2 \log(m \log(3m/2)^2 + m \log(3m))] \\
     & =   & m \log(c \Delta) + m/2 \log[m \log(c \Delta) + m/2 \log(2 m \log(3m/2) + m \log(3m))] \\
     & \le & m \log(c \Delta) + m/2 \log[m \log(c \Delta) + m/2 \log(3 m \log(3m))] \\
     & \le & m \log(c \Delta) + m/2 \log[2m \log(3m/2) + m/2 \log(3 m \log(3m))] \\
     & = &   m \log(c \Delta) + m/2 \log[2m \log(3m/2) + m \log(3 m \log(3m))^{1/2}] \\
     & \le &  m \log(c \Delta) + m/2 \log[2m \log(3m/2) + m \log(3 m)] \\
     & \le &  m \log(c \Delta) + m/2 \log[3m \log(3m)] \\
     & =   &  m \log(c \Delta \sqrt{3m \log(3m)})
     \end{array} $$
     
Combining all cases we obtain the following first result:

\begin{theorem}
\label{thm:ineq:main}
The smallest value of $|Y|$ fulfilling Equation~\eqref{log:inequality} can be
bounded as follows: 
\begin{center}
\begin{tabular}{|l||l|}
\hline
 upper bound & condition on $m, c \Delta$  \\
\hline
$m \log(c \Delta \sqrt{3m/2 \log(c \Delta)})$  & $m < 2/3 {c \Delta \over \log(c \Delta)}$ \\
\hline
$m \log(c \Delta \sqrt{3m \log(3m/\sqrt{2})})$ & $m \ge 2/3 {c \Delta \over \log(c \Delta)}$ \\
 &  and $\log(c \Delta) < (c \Delta)^{1/2}$ \\
\hline
$m \log(c \Delta \sqrt{3m/2 \log(3m \log(c \Delta))})$ &  $m \ge 2/3 {c \Delta \over \log(c \Delta)}$ \\
 &  and $(c \Delta)^{1/2} \ge 3m/2$ \\
\hline 
$m \log(c \Delta \sqrt{3m \log(3m)})$ & $m \ge 2/3 {c \Delta \over \log(c \Delta)}$ \\
 & and $\log(c \Delta) \ge (c \Delta)^{1/2}$ \\
 & and $(c \Delta)^{1/2} < 3m/2$ \\
\hline
\end{tabular} 
\end{center} 
\end{theorem}

\subsection{Improved bounds}

The next step is to prove natural bounds for the cardinality $|Y|$. We conjecture that we get
$\le \alpha m \log(c' \Delta m^{1/2})$ with $\alpha$ close to $1$. The first
step is to check in which cases we get $\le 3/2 m \log(c \Delta m^{1/2})$.

In some cases we have either to increase $c$ to $c'$ or to modify the analysis above for some cases.

{\bf Case 1} $m < 2/3 {c \Delta \over \log(c \Delta)}$. The bound from the previous subsection can be bounded, if
  $$|Y| \le m \log(c \Delta \sqrt{3m/2 \log(c \Delta)})
\le (3/2) m \log(c \Delta m^{1/2}),$$ which is equivalent to
$$ c \Delta \sqrt{3m/2 \log(c \Delta)} \le (c \Delta m^{1/2})^{3/2} = (c \Delta)^{3/2} m^{3/4}$$
This can be transformed into
$$ (3m/2) \log(c \Delta) \le c \Delta m^{3/2}$$ or equivalently
$$ (3/2) \log(c \Delta) \le c \Delta m^{1/2}.$$
The property above $m < 2/3 {c \Delta \over \log(c \Delta)}$ gives
$$ (3/2) \log(c \Delta) < c \Delta / m \le c \Delta m^{1/2}$$
which holds for any $m \ge 1$ (using $m^{3/2} \ge 1$).

{\bf Case 2:} $c \Delta < 3m/2 \log(c \Delta)$ and $\log(c \Delta) < (c \Delta)^{1/2}$. The upper bound from the previous subsection can be estimated as follows. We have
$$|Y| \le m \log(c \Delta \sqrt{3m \log(3m/\sqrt{2})}  \le 3/2 m \log(c' \Delta m^{1/2}),$$ if
$$ c \Delta \sqrt{3m \log(3m/\sqrt{2})} \le (c' \Delta)^{3/2} m^{3/4}.$$
This is equivalent to
$$ 3m \log(3m/\sqrt{2}) \le c' (c'/c)^2 \Delta m^{3/2}$$
or
$$ 3 \log(3m/\sqrt{2}) \le c' (c'/c)^2 \Delta m^{1/2}.$$

First we consider the case $c \Delta \ge 4$. Here we could use also our previous better bound with $2$ in the log term instead of $\sqrt{2}$. Since $\log(c \Delta) < (c \Delta)^{1/2}$, we know that actually $c \Delta > 16$ and obtain using Lemma \ref{log:bound1}
$$ \log(3m/2 \log(c \Delta)) < (3m/2 \log(c \Delta))^{1/2}.$$
This implies $\log(3m/2) < (3m/2)^{1/2} (\log(c \Delta))^{1/2} - \log(\log(c \Delta))$. To prove the inequality $3 \log(3m/2) \le c \Delta m^{1/2}$, we study the inequality $ 3 (3m/2)^{1/2} (\log(c \Delta))^{1/2} - 3 \log(\log(c \Delta)) \le c \Delta m^{1/2}$. The inequality holds, when
$ 3 \sqrt{3/2}  (c \Delta)^{1/4} \le c \Delta $ or $(c \Delta)^{3/4} \ge
3 \sqrt{3/2}$. This is equivalent to $c \Delta \ge (3 \sqrt{3/2})^{4/3} =
((3 \sqrt{3/2})^4)^{1/3} = (3^6/4)^{1/3} = 3^2/(4)^{1/3}$. Since the right hand term is at most $9$ and $c \Delta \ge 16$, the inequality holds for $c \Delta \ge 4$.

The remaining case is $c \Delta < 4$. Here we need the bound $3 \log(3m/\sqrt{2}) \le c' (c'/c)^2 \Delta m^{1/2} $. The bounds holds directly, when $m$ is large enough such that $3 \log(3m/\sqrt{2}) \le c \Delta m^{1/2}$. Using $c \Delta \ge 2$, this implies $3/2 \log(3/\sqrt{2}) +
3/2 \log(m) \le m^{1/2}$ or $1.628 + 3/2 \log(m) \le m^{1/2}$. This holds for
example whenever $m \ge 256$. Another possibility is to check when
$3 \log(3/\sqrt{2}) + 3 \log(m) \le  (c'/c)^3 c \Delta m^{1/2}$. For $m \le 256$ we could use $c' = 2.3514 c$ and obtain $3.256 + 3 \log(m) \le 26 \le (c'/c)^3 c \Delta m^{1/2} $ for $c \Delta \ge 2$.

{\bf Case 3:} In this case we have only to consider $m=2$ and $c \Delta \in [9,10]$. Here we have to bound $|Y| = m \log(c \Delta \log(2 \log^2(c \Delta))) \le 3/2 m \log(c \Delta m^{1/2})$. The inequality holds when
$$ c \Delta \log(2 \log^2(c \Delta)) \le (c \Delta m^{1/2})^{3/2}$$ or equivalently when
$$  \log^2(2 \log^2(c \Delta)) \le (c \Delta) m^{3/2}.$$
In this case we have $(c \Delta) m^{3/2} \ge 9 \sqrt{8} \approx 25.45$ and
$\log^2(2 \log^2(\Delta)) \le \log^2(2 \log^2(10)) \le 19.93$. Therefore, the
inequality is true.

{\bf Case 4:} Here we have to show that $|Y| = m \log(c \Delta (3 m \log(3m))^{1/2})
\le (3/2 m) \log(c' \Delta m^{1/2})$. As properties we can use $c \Delta < (3m/2) \log(c \Delta)$, $\log(c \Delta) \ge (c \Delta)^{1/2}$ and $(c \Delta)^{1/2} \le 3m/2$. The inequality above holds (similar to case 2), whenever
$3 \log(3m) \le c' (c'/c)^2 \Delta m^{1/2}$. Notice that $\log(c \Delta) \ge (c \Delta)^{1/2}$ holds only for $c \Delta \in [4,16]$; and this further implies $m \ge 2$. For $c' = c$ the inequality $3 \log(3m) \le 4 m^{1/2}$
is equivalent to $3/4 \log(3) + 3/4 \log(m) \le 1.18873 + 3/4 \log(m) \le m^{1/2}$ and this holds for
$m \ge 20$. In the remaining we may assume that $m \in \{2,\ldots, 19\}$. We study now
$3 \log(3m) \le (c'/c)^3 (c \Delta) m^{1/2}$ or
$1.18873 + 3/4 \log(m) \le (c'/c)^3 m^{1/2}$. We can bound
$1.18873 + 3/4 \log(m) \le  1.18873 + 3/4 \log(19) \le 4.37466 \le (c'/c)^3 \sqrt{2}$. The last inequality actually holds for $c' \ge 1.46 c$.

In total we obtain the following result:

\begin{theorem}
The smallest $|Y|$ fulfilling Equation~\eqref{log:inequality} is at most 
$$ (3/2) m \log(c' \Delta m^{1/2}) $$
where $c' \approx 2.3514 c$. 
\end{theorem}

\end{document}